\documentclass[12pt]{article}
\usepackage[textwidth=5.5in,top=1.37in,bottom=1.37in]{geometry}
\usepackage{amsthm,amsmath,amssymb,bbm,bm}
\usepackage{natbib}
\usepackage{multirow}
\usepackage[pdftex]{graphicx}
\usepackage{subfigure}
\usepackage{wrapfig}
\usepackage{tikz}
\usepackage{centernot}
\usepackage{array}
\usepackage{url}
\usepackage{float}
\usepackage{algorithmic}
\usepackage[linesnumbered, ruled, vlined]{algorithm2e}
\usepackage{mathrsfs}
\usepackage{dsfont}
\usepackage{titling}
\usepackage{relsize}
\usepackage{rotating}
\usepackage{enumitem}
\usepackage{setspace}

\usepackage[colorlinks=true,
            linkcolor=red,
            anchorcolor=blue,
            citecolor=blue,
            urlcolor=blue]{hyperref}

\newtheorem{remark}{Remark}
\newtheorem{assumption}{Assumption}

\newtheorem{theorem}{Theorem}[section]
\newtheorem{lemma}{Lemma}[section] 

\newcommand{\expit}{{\text{expit}}}
\newcommand{\logit}{{\text{logit}}}
\newcommand{\sign}{{\text{sign}}}

\newcommand{\cV}{{\mathcal{V}}}
\newcommand{\cM}{{\mathcal{M}}}
\newcommand{\cD}{{\mathcal{D}}}
\newcommand{\cL}{{\mathcal{L}}}
\newcommand{\cH}{{\mathcal{H}}}
\newcommand{\cG}{{\mathcal{G}}}
\newcommand{\PP}{{\mathbbm{P}}}
\newcommand{\cO}{{\mathcal{O}}}

\newcommand{\indep}{\rotatebox[origin=c]{90}{$\models$}}
\newcommand{\yifan}[1]{\textcolor{black}{#1}\xspace}

\newcommand{\source}[1]{#1}
\newcommand*\samethanks[1][\value{footnote}]{\footnotemark[#1]}

\begin{document}
\onehalfspacing
\title{\vspace{-1.5cm}  {\Large A semiparametric instrumental variable approach to optimal treatment regimes under endogeneity}}

\author{
\large Yifan Cui\thanks{Department of Statistics, The Wharton School, University of Pennsylvania, Philadelphia, PA 19104}~\thanks{Email: \href{mailto:cuiy@wharton.upenn.edu}{cuiy@wharton.upenn.edu}}
\and \large Eric Tchetgen Tchetgen\samethanks[1]}
\predate{}
\postdate{}
\date{}


\maketitle
\vspace{-0.5cm}

\begin{abstract}
 There is a fast-growing literature on estimating optimal \mbox{treatment} regimes based on randomized trials or observational studies under a key identifying condition of no unmeasured confounding.
Because confounding by unmeasured factors cannot generally be ruled out with certainty in observational studies or randomized trials subject to non-compliance, we propose a general instrumental variable approach to learning optimal treatment regimes under endogeneity.
Specifically, we establish identification of both value function $E[Y_{\cD(L)}]$ for a given regime $\cD$ and optimal regimes $\arg \max_{\cD} E[Y_{\cD(L)}]$  with the aid of a binary instrumental variable, when no unmeasured confounding fails to hold.
We also \mbox{construct} novel multiply robust classification-based estimators.
Furthermore, we propose to identify and estimate optimal treatment regimes among those who would comply to the assigned treatment under a monotonicity assumption. In this latter case,  we establish the somewhat surprising result that complier optimal regimes can be consistently estimated without directly collecting compliance information and therefore without the complier average treatment effect itself being identified. Our approach is illustrated via extensive simulation studies and a data application on the effect of child rearing on labor participation.
\end{abstract}

\noindent {\bf keywords}
Precision medicine, Optimal treatment regimes, Complier optimal regimes, Instrumental variable, Unmeasured confounding

\section{Introduction}\label{sec:intro}
The primary goal of estimating an individualized treatment regime is to recover a rule which assigns the treatment, among a set of possible treatments, to each patient based on the individual's characteristics.
Optimal treatment regimes have recently received a lot of attention in \yifan{the statistical and biomedical literatures}. A prevailing strand of work in this literature has approached the optimal treatment problem through either Q-learning \citep{chakraborty2010dtr,qian2011performance,laber2014q,schulte2014} or A-learning \citep{Robins2000MarginalSM,murphy2003optimal,Robins2004,shi2018}. Recently, an alternative approach has emerged from a classification perspective \citep{zhang2012estimating,zhao2012estimating,rubin2012statistical}, \yifan{which has proven more robust to model misspecification in some settings}.

Recent explorations of optimal individual treatment regimes have considered a variety of data types, including within the context of standard randomized experiments \citep{kosorok2016,kosorok2019review,tsiatis2019dynamic}, but also observational studies \citep{athey2017efficient,kallus2018balanced} and electronic health records \citep{Wang2016LearningOI,wu2018ehr}. There has also been work on individualized treatment regimes with a somewhat different objective, \yifan{such as policy improvement \citep{kallus2018confounding}, quantiles for outcome measure \citep{linn2015,linn2017,wang2018quantile}}, tail control \citep{qi2019siam}, and interpretability \citep{Orellana2010pm,Laber2015tree,Zhang2015decisionlist}.

A common assumption made in prior work on estimating optimal treatment regimes is that of no unmeasured confounding; an assumption which cannot be guaranteed in observational studies, nor in randomized experiments subject to non-compliance. Without such an assumption, it is well-known that causal effects, and in particular optimal treatment regimes cannot be identified nonparametrically without an alternative assumption.
The central controversy of the unconfoundedness assumption is that one typically needs to collect and appropriately account for a large number of relevant covariates in order to make the assumption credible. \yifan{The use of instrumental variables (IVs)} is a well-known approach to estimating causal effects in observational studies or randomized trials with non-compliance.
An IV is defined as a \mbox{pretreatment} variable that is independent of all unmeasured confounders, and does not have a direct causal effect on the outcome other than through the treatment.
In a double-blind, placebo-controlled, randomized trial, random assignment is a common example of an ideal IV when patients fail to comply to assigned treatment.
A prominent IV approach in epidemiological studies, known as Mendelian randomization studies, leverages genetics variants known to be associated with the phenotype defining the exposure, in order to estimate the causal effect of the phenotype on a health outcome.  A well-known illustration of the approach takes fat mass and obesity-associated protein (FTO) as a genetic IV to estimate a causal association between body mass index (BMI) and depression \citep{Walter2015}.

\cite{imbens1994,angrist1996} proposed a formal counterfactual based approach for binary treatment and IV, and established identification under a certain monotonicity assumption, of the so-called complier average causal effect, i.e., the average treatment effect for the subset of the population who would always comply to their assigned treatment.
Building on this original work, \cite{ABADIE2003231,tan2006,ogburn2015} have developed various semiparametric methods for estimating complier average treatment effects with appealing robustness and efficiency properties.
The average treatment effect generally differs from complier average causal effect and often is the causal effect of primary interest  \citep{hernan2006epi,aronow2013,wang2018bounded}.
\cite{wang2018bounded} formally established identification of the population average treatment effect under certain no-interaction assumptions.
Although it has developed a rich literature on IV methods for static regimes, to the best of our knowledge, no prior literature exists on IV methods for optimal treatment regimes.

In this paper, we propose a number of IV learning methods for estimating optimal treatment regimes in case no unmeasured confounder assumption fails to hold.
Specifically, we adapt and extend the weighted classification perspective pioneered by \cite{zhang2012estimating,zhao2012estimating,rubin2012statistical}, by allowing for an endogenous treatment (i.e., confounded by unmeasured factors) which we account for by a novel use of an IV. We take a classification perspective as it is now widely recognized to be quite versatile for the purpose of estimating optimal treatment regimes,  because of the large arsenal of robust classification methods and corresponding off-the-shelf software that one can readily leverage.

The paper makes a number of contributions to both the IV and precision medicine literatures.  First, we establish identification of optimal treatment regimes for a binary treatment subject to unmeasured confounding, by leveraging a binary IV.
\yifan{The proposed identification conditions give rise to IV estimators of optimal treatment regimes without necessarily identifying the value function for a given regime.}
In addition, we construct multiply robust classification-based estimators of optimal treatment regimes provided that a subset of several posited models indexing the observed data distribution is correctly specified.
Second, we propose to identify and estimate optimal treatment regimes among the subset of the population that would always comply to their assigned treatment.
A somewhat surprising theorem establishes that under the identifying assumption of  no defier (i.e., monotonicity of the effect of the IV on the treatment) one can in fact identify complier optimal treatment regimes even when individuals' realized treatment values are not observed, and therefore the complier average treatment effect itself is not \mbox{identifiable}. Our results therefore imply that in a randomized trial subject to non-compliance, it is possible to consistently infer optimal treatment regimes for compliers even if, as often the case in practice, investigators fail to collect adherence information on individual participants. For instance, in a randomized trial with one-sided non-compliance (i.e., when the placebo group cannot access the experimental treatment) whereby monotonicity holds by design, one can obtain assumption-free inferences on who might benefit from the intervention, without necessarily knowing who in the treatment arm adhered to the assigned treatment, and therefore one cannot recover the actual magnitude of the treatment effect among compliers.

\yifan{Our simulation studies confirm that the proposed inverse weighted and multiply robust estimators perform well in a range of settings and in fact outperform existing methods in settings where unmeasured confounding is strong. In particular, the proposed estimators have significantly higher empirical value function, i.e., higher average potential outcome under the estimated optimal treatment regime in the presence of unmeasured confounding. In addition, the performance of the proposed estimators is comparable to that of prior methods when there is no unmeasured confounding.
We also apply the proposed methods to a data application on the effect of child rearing on labor participation.}

The remainder of the article is organized as follows. In Section~\ref{sec:method}, we present the mathematical framework for the use of IVs in estimating individualized treatment regimes subject to unmeasured confounding.
\yifan{Section~\ref{sec:robustness} develops two novel multiply robust classification-based estimators.}
Extensive simulation studies are presented in Section~\ref{sec:simulation}. Section~\ref{sec:realdata} describes application of the proposed methods to mother's labor participation.
Next, we propose to identify optimal treatment regimes among those who would always comply to their assigned treatment in Section~\ref{sec:cace}.
The article concludes with a discussion of future work in Section~\ref{sec:discussion}. Proofs and additional results are provided in Appendix and \yifan{Supplementary Material}.

\section{Methodology}\label{sec:method}

We briefly introduce some general notation used throughout the paper. Let $Y$ denote the outcome of interest and $A \in \{+1,-1\}$ be a binary treatment indicator.
Suppose that $U$ is an unmeasured confounder (possibly vector-value) of the effect of $A$ on $Y$.  Suppose also that one has observed a pretreatment binary  instrumental variable $Z \in \{+1,-1\}$. Let $L\in \mathcal L$ denote a set of fully observed pre-IV covariates, where $\mathcal L$ is a $p$-dimensional vector space.
Throughout we assume the complete data are independent and identically distributed realizations of $(Y, L, A, U, Z)$; thus the observed data are $\cO=(Y,L,A,Z)$.

We wish to identify a treatment regime $\cD$, which is a mapping from the patient-level covariate space $\cL$ to the treatment
space $\{+1, -1\}$ that maximizes the corresponding expected potential outcome for the population. In other words, the goal is to estimate an optimal treatment regime, defined as follows,
\begin{align}
\cD^*(L) = \sign\{ E(Y_1-Y_{-1}|L) \},
\label{eq:opt2}
\end{align}
where $Y_a$  is a person's potential outcome under an intervention that sets treatment to value $a$, $\sign(x) = 1$ if $x > 0$ and $\sign(x) = -1$ if $x < 0$.
Throughout it is assumed that larger values of $Y$ are more desirable.

Let $Y_{\cD(L)}$ be the potential outcome under a hypothetical intervention that assigns treatment according to regime $\cD$; this potential outcome is equivalently expressed as
\begin{align*}
Y_{\cD(L)} \equiv Y_{1}I\{\cD(L)=1\}+Y_{-1}I\{\cD(L)=-1\},
\end{align*}
where $I\{\cdot\}$ is the indicator function.
\yifan{Throughout the paper, we make the following standard consistency and positivity assumptions: (i) For a given regime $\cD$, $Y=Y_{\cD(L)}$ when $A=\cD(L)$ almost surely. That is, a person's observed outcome matches his/her potential outcome under a given treatment regime when the realized treatment matches his/her potential treatment assignment under the regime; (ii)  We assume that $\Pr(A=a|L)>0$ for $a=\pm 1$ almost surely, i.e., a person has an opportunity to receive both treatments.}

\subsection{Optimal treatment regimes subject to no unmeasured confounding}
Prior methods for estimating optimal treatment regimes have typically relied on the following unconfoundedness assumption:
\begin{assumption}\emph{(Unconfoundedness)}
$Y_a \indep A| L$ for $a=\pm 1$.
\label{asm:unconfoundedness}
\end{assumption}
The assumption essentially rules out the existence of an unmeasured factor $U$ that confounds the effects of $A$ on $Y$ upon conditioning on $L$. Such an assumption is untestable without further  restriction on the data generating mechanism and cannot generally be enforced outside of an ideal randomized study.

It is straightforward to verify that
under Assumption~\ref{asm:unconfoundedness}, one can identify the counterfactual mean (known as the value function of regime $\cD$, \citep{qian2011performance}) $E[Y_{\cD(L)}]$ for a given treatment regime $\cD$.
Furthermore, optimal treatment regimes in Equation~\eqref{eq:opt2} are identified from the observed data by the following expression,
\begin{align*}
\cD^*(L) = \sign\{ E(Y|L,A=1) - E(Y|L,A=-1) \}.
\end{align*}
As established by \cite{qian2011performance}, learning optimal individualized treatment regimes under unconfoundedness can alternatively be formulated as
\begin{align}
\cD^*=\arg\max_{\cD} E_L \left[  E_{Y_{\cD}} [Y_{\cD(L)}|L] \right]=\arg\max_{\cD} E\left[\frac{I\{A=\cD(L)\}Y}{f(A|L)}\right],
\label{eq:opt1}
\end{align}
\cite{zhang2012robust} proposed to directly maximize the value function over a restricted set of functions.

Rather than maximizing the above value function, \cite{zhao2012estimating,zhang2012estimating} transformed the above problem into a formal, equivalent weighted classification problem,
\begin{align}
\cD^*=\arg\min_{\cD} E\left[\frac{Y}{f(A|L)}I\{A\neq \cD(L)\}\right],
\label{eq:opt3}
\end{align}
with  0-1 loss function and weight ${Y}/{f(A|L)}$. \cite{zhao2012estimating} addressed the computational burden of formulation \eqref{eq:opt3} by substituting the 0-1 loss with the hinge loss and proposed to solve the optimization via support vector machines. The ensuing classification approach was shown to have appealing robustness properties, particularly in the context of a randomized study where no model assumption is needed.

Subsequent work has provided further extensions and refinements of the classification perspective
\citep{zhang2012estimating,zhao2015new,zhao2015doubly, chen2016dosefinding,rwl,zhou2017augmented,cui2017tree,zhu2017greedy,Liu2018AugmentedOL,zhangzhang2018}.
Notably, all prior methods, whether classification-based or not, rely on the unconfoundedness Assumption~\ref{asm:unconfoundedness}. As the assumption may not hold in observational studies or randomized trials with non-compliance, in the next section, we introduce a general framework for learning optimal treatment regimes under endogeneity (i.e., unmeasured confounding).

\subsection{Identification of optimal treatment regimes with unmeasured confounding}
In this section, we no longer rely on Assumption~\ref{asm:unconfoundedness} and therefore allow for unmeasured confounding. Instead, let $Y_{z,a}$ denote the potential outcome had, possibly contrary to fact, a person's IV and treatment value been set to $z$ and $a$, respectively. Suppose that the following assumption holds.
\begin{assumption}\emph{(Latent unconfoundedness)}
$Y_{z,a} \indep (Z,A) |L,U$ for $z,a=\pm 1$.
\label{asm:unconfoundedness2}
\end{assumption}
\vspace{-0.6cm}
This assumption essentially states that together $L$ and $U$ would in principle suffice to account for confounding of the joint effect of $Z$ and $A$ on $Y$.
Because $U$ is not observed, we propose to account for it by making the following standard IV assumptions:

\begin{assumption}\emph{(IV relevance)} $Z \centernot{\indep} A|L$.
\label{IV Relevance}
\end{assumption}

\begin{assumption}\emph{(Exclusion restriction)} $Y_{z,a}=Y_a$ for $z,a=\pm 1$ almost surely.
\label{Exclusion Restriction}
\end{assumption}

\vspace{-0.95cm}

\begin{assumption}\emph{(IV independence)} $Z \indep U|L$.
\label{IV Independence}
\end{assumption}

\begin{assumption}\emph{(IV positivity)} $0<f\left(  Z=1|L\right)<1$ almost surely.
\label{asm:IV positivity}
\end{assumption}

The first three conditions are well-known core IV conditions, while Assumption~\ref{asm:IV positivity} is needed for nonparametric identification \citep{10.1093/ije/29.4.722,hernan2006epi}.
Assumption~\ref{IV Relevance} requires that the IV is associated with the treatment conditional on $L$.  In a placebo controlled randomized trial with non-compliance, this assumption will typically be satisfied for \mbox{treatment} assignment $Z$ and treatment as taken $A$ whenever more individuals take the active treatment in the intervention arm than in the placebo arm.  Note that Assumption~\ref{IV Relevance} does not rule out confounding of the $Z$-$A$ association by an unmeasured factor, however, if present, such factor must be independent of $U$.  We will refer to $Z$ as a causal IV in case no such confounding is present.
 Assumption~\ref{Exclusion Restriction} states that there can be no direct causal effect of $Z$ on $Y$ not mediated by $A$.
 Assumption~\ref{IV Independence} ensures that the causal effect of $Z$ on $Y$ is unconfounded given $L$.  Figure~\ref{fig:sra dag} provides a graphical representation of Assumptions~\ref{Exclusion Restriction} and \ref{IV Independence} for a causal IV.

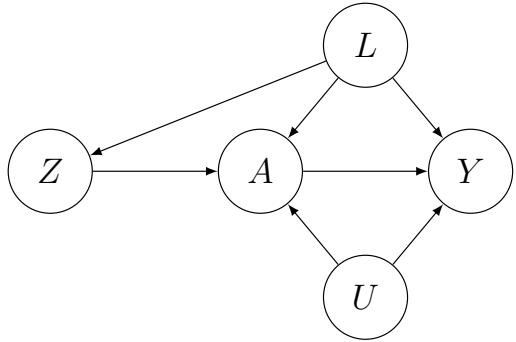
\begin{figure}[h]
  \centering
  \vfill
  \resizebox{200pt}{!}{%
    \begin{tikzpicture}[state/.style={circle, draw, minimum size=1cm}]
  \def\Ax{0}
  \def\Ay{0}
  \def\offset{2.5}
  \def\Bx{\Ax+5}
  \def\By{\Ay}
  \node[state,shape=circle,draw=black] (Z) at (\Ax,\Ay) {$Z$};
  \node[state,shape=circle,draw=black] (Y) at (\Bx,\By) {$Y$};
  \node[state,shape=circle,draw=black] (A) at (\Bx-\offset,\By) {$A$};
  \node[state,shape=circle,draw=black] (L) at (\Bx-1.25,\By+1.5) {$L$};
  \node[state,shape=circle,draw=black] (U) at (\Bx-1.25,\By-1.5) {$U$};

  \draw [-latex] (Z) to [bend left=0] (A);
  \draw [-latex] (A) to [bend left=0] (Y);
  \draw [-latex] (L) to [bend left=0] (A);
  \draw [-latex] (L) to [bend left=0] (Z);
  \draw [-latex] (L) to [bend left=0] (Y);
  \draw [-latex] (U) to [bend left=0] (A);
  \draw [-latex] (U) to [bend left=0] (Y);

\end{tikzpicture}
  }
  \vfill
  \caption{Causal DAG with unmeasured confounding and a causal IV.}
  \label{fig:sra dag}
\end{figure}

Under Assumptions~\ref{asm:unconfoundedness2}-\ref{IV Independence}, it is not possible to uniquely identify the value function of a given regime $\cD$.
Thus, directly optimizing the value function does not appear possible even with a valid IV encoded in the DAG of Figure~\ref{fig:sra dag}. Nonetheless, it is possible to identify treatment regimes that maximize lower bounds of the value function (see Section~\ref{bounds} in the Supplementary Material) without an additional assumption.
In order to identify $\cD^*$, we consider the following assumption.
\begin{assumption}\emph{(No unmeasured common effect modifier)}
$$ Cov\left\{\widetilde \delta(L,U), \widetilde \gamma(L,U)|L\right\} = 0,$$ almost surely,
where $\widetilde \delta(L,U) \equiv \Pr(A=1|Z=1,L,U)-\Pr(A=1|Z=-1,L,U)$ and $\widetilde \gamma(L,U) \equiv E(Y_1-Y_{-1}|L,U)$, respectively.
\label{asm:weak}
\end{assumption}

Assumption~\ref{asm:weak} essentially states that there is no common effect modifier by an unmeasured confounder, of the additive effect of treatment on the outcome, and the additive effect of the IV on treatment.
We also consider the following stronger condition.
\begin{assumption}\emph{(Independent compliance type)}
$$ \delta(L) \equiv \Pr(A=1 | Z=1,L)-\Pr(A=1| Z=-1,L)=\widetilde \delta(L,U)~ \text{almost surely}.$$
\label{asm:strong}
\end{assumption}
\vspace{-1cm}

Assumption~\ref{asm:strong} essentially states that there is no additive interaction between $Z$ and $U$ in a model for the probability of being treated conditional on $L$ and $U$. As stated in \cite{wang2018bounded}, this assumption would hold if $U$ was independent of a person's compliance type. Assumption~\ref{asm:weak} is implied by and therefore more general than Assumption A5 of \cite{wang2018bounded} which states that either $\widetilde \delta(L,U)$ or $\widetilde \gamma(L,U)$ does not depend on $U$ \citep{wang2018}. Clearly, Assumption~\ref{asm:strong} implies Assumption~\ref{asm:weak} and therefore is more stringent.
Now we are ready to state our first identification result.
\begin{theorem}
Under Assumptions~\ref{asm:unconfoundedness2}-\ref{asm:weak}, $\arg \max_\cD E[Y_{\cD(L)}]$ is nonparametrically identified,
\begin{align}
\arg\max_\cD E[Y_{\cD(L)}] = \arg\max_{\cD}E \left[\frac{ZAYI\{A=\cD(L)\}}{\delta(L)f(Z|L) } \right].
\label{eq:identification1}
\end{align}
Furthermore, under Assumptions~\ref{asm:unconfoundedness2}-\ref{asm:IV positivity} and \ref{asm:strong}, for a given regime $\cD$,
\begin{align}
E[Y_{\cD(L)}] =  \cV(\cD) \equiv E \left[\frac{ZAYI\{A=\cD(L)\}}{\delta(L)f(Z|L) } \right].
\label{eq:vd}
\end{align}
\label{thm:identification}
\end{theorem}
Theorem~\ref{thm:identification} gives one of our main identification results, and states that optimal treatment regimes are nonparametrically identified with a valid IV satisfying Assumption~\ref{asm:weak}, therefore extending prior identification of optimal  treatment regimes to account for potential confounding by an unmeasured factor.
The theorem further states that \yifan{the functional} $\cV(\cD)$ nonparametrically identifies the value function with a valid IV under the stronger Assumption~\ref{asm:strong}.
Theorem~\ref{thm:identification} also serves as basis for the estimator proposed in the next section.

The following theorem shows that progress can be made towards identifying optimal treatment regimes without necessarily using a person's realized treatment value $A$.
\begin{theorem}
Under Assumptions~\ref{asm:unconfoundedness2}-\ref{asm:weak},
\begin{align}
\arg \max_\cD E[Y_{\cD(L)}] = \arg\max_{\cD} E \left[\frac{YI\{Z=\cD(L)\}}{\delta(L)f(Z|L) } \right].
\label{eq:identification2}
\end{align}
\label{thm:identification2}
\end{theorem}

\vspace{-1cm}
\begin{remark}
Interestingly, Theorem~\ref{thm:identification2} implies that if it is known a priori that the association between $Z$ and $A$ is positive conditional on $L$, i.e., $\delta(L)>0$, then it is possible to identify optimal treatment regimes even if one does not directly observe the treatment variable $A$, by solving the optimization problem $\arg\max_{\cD} E[\widetilde W(L) I\{Z=\cD(L)\}Y/f(Z|L)]$ for any choice of weight $\widetilde W(L)>0$. Furthermore, in the event that external information is available on $\delta(L)$, as would sometime be the case if a separate sample with data on $A$, $Z$ and $L$, solving the above optimization problem with $\widetilde W(L)=1/\delta(L)$ recovers Equation~\eqref{eq:identification2}.
\label{remark1}
\end{remark}

\subsection{An IV approach to weighted learning}\label{sec:2.3}

In this section, motivated by Theorems~\ref{thm:identification} and \ref{thm:identification2}, we propose two classification-based estimators.
To further motivate our classification-based approach, note that
optimization tasks \eqref{eq:identification1} and \eqref{eq:identification2} are equivalent to
 \begin{align}
\arg\min_{\cD} E \left[ W^{(1)} I\{A\neq \cD(L)\} \right], \label{eq:optidentification1}\\ \arg\min_{\cD} E \left[ W^{(2)} I\{Z\neq \cD(L)\} \right], \label{eq:optidentification2}
\end{align}
respectively, where
\begin{align*}
W^{(1)}=\frac{ZAY}{\delta(L)f(Z|L) },\quad W^{(2)}=\frac{Y}{\delta(L)f(Z|L) }.
\end{align*}
In the rest of this section, we focus primarily on Equation~\eqref{eq:optidentification1} to develop our estimator although our results readily extend to Equation~\eqref{eq:optidentification2}.
\yifan{The idea behind our estimator is similar in spirit to \cite{zhang2012estimating,zhao2012estimating,rubin2012statistical,zhangzhang2018} in that} this alternative formulation of the optimization task may be interpreted as a classification problem in which one aims to classify $A$ using $L$ with misclassification error given by the weighted outcome $W$. Because the weight may not be positive (e.g., for samples with $A \neq Z$ even if $Y/\delta(L)>0$), in order to remedy this difficulty, we further modify the weights  by adopting the approach of \cite{liu2016hybrid} and leverage the following equality
\begin{align}\label{eq:equiv}
\arg\min_\cD E \left[ |W| I\{ \sign(W)A \neq \cD(L)\} \right]= \arg\min_\cD E \left[ W I\{A\neq \cD(L)\} \right].
\end{align}
\yifan{We follow \cite{zhao2012estimating} and proceed with convex optimization via the use of hinge loss function.
Furthermore, we penalize the complexity of the decision function to avoid overfitting. Thus, we propose to estimate optimal treatment regimes by minimizing the following regularized objective function,
\begin{align}
\label{eq:obj}
\widehat g = \arg \min_{g} \frac{1}{n} \sum_{i=1}^n |W_i| \phi \big( \sign(W_i) A_i g(L_i) \big) +\frac{\lambda}{2}||g||^2,
\end{align}
where $\phi$ is the hinge loss function, $g$ encodes the decision function within a specific class, $\lambda$ is a tuning parameter, and $\widehat \cD = \sign(\widehat g)$. The weight $W_i$ is unknown and therefore estimated from the data in a first step,
 and then substituted in \eqref{eq:obj}. We denote $$\widehat W^{(1)}=\frac{Z A Y}{\widehat \delta(L) \widehat f(Z|L)},\widehat W^{(2)}=\frac{Y}{\widehat \delta(L) \widehat f(Z|L)},$$ where $\widehat \delta (L)$ and $\widehat f(Z|L)$ can be fitted by a logistic regression or a nonparametric model such as random forest.}

We refer to \cite{zhao2012estimating} for solving this optimization with linear and nonlinear decision rules.
Note that variable selection techniques can be incorporated in the proposed approach when the dimension of covariates is high. For example, the variable selection techniques developed in \cite{zhao2012estimating,rwl} can readily be adopted here,  in which $l_2$ penalty is replaced with the elastic net penalty \citep{zou2005}.
The Fisher consistency, excess risk bound and universal consistency of the estimated treatment regime are shown in the Supplementary Material. The proof is akin to \cite{zhao2012estimating,zhou2017augmented}.

\section{Multiply robust classification-based estimators}\label{sec:robustness}
\yifan{We characterize the efficient influence function of $\cV(\cD)$ under Assumptions \ref{asm:unconfoundedness2}-\ref{asm:IV positivity} and \ref{asm:strong} as shown in Theorem~\ref{thm:triply1} in the Supplementary Material.
In principle, one could estimate optimal treatment regimes which maximize $E[Y_{\cD}(L)]$ over a class of parametric models of regimes, as proposed by \cite{zhang2012robust} in case of unconfoundedness.
However, this kind of approach may restrict treatment regimes to a relatively small set of possible functions and therefore may be suboptimal. An alternative approach is to develop a multiply robust classification-based estimator \citep{zhang2012estimating,zhao2012estimating,rubin2012statistical,zhangzhang2018} of optimal treatment regimes as described in the following.}

We denote the conditional average treatment effect as
$$\Delta(L)= E(Y_1-Y_{-1}|L).$$
Note that in order to learn optimal treatment regimes, we essentially need to minimize the following weighted classification error with respect to regime $\cD$,
\begin{align}
E[A \Delta(L) I\{A \neq \cD(L)\}] \quad \text{or} \quad E[Z \Delta(L) I\{Z \neq \cD(L)\}].
\label{eq:substitue}
\end{align}
The choice of statistic
\begin{align*}
\widetilde \Delta(L) = \frac{ZY}{\delta(L)f(Z|L)},
\end{align*}
substituted in Equation~\eqref{eq:substitue} recovers the estimators which we described in Section~\ref{sec:method}.
Motivated by the form of the efficient influence function of $E[\Delta(L)]$ \citep{wang2018bounded}, we propose the following statistic
\begin{align*}
\widetilde \Delta_{MR}(L) = \frac{Z}{\delta(L)f(Z|L) }\bigg [Y -A \Delta(L) -E[Y|Z=-1,L] +\Delta(L) E[A|Z=-1,L] \bigg ]+ \Delta(L)
\end{align*}
to obtain the following multiply robust weights
\begin{align*}
\widehat W^{(1)}_{MR} =&\bigg[ \frac{ZA}{ \delta(L,\widehat \beta) \widehat f(Z|L) }\bigg \{Y -A  \Delta(L, \widehat \theta) \\ & -\widehat E[Y|Z=-1,L] + \Delta(L, \widehat \theta) \widehat E[A|Z=-1,L] \bigg \}+ A\Delta(L, \widehat \theta) \bigg],
\end{align*}
and
\begin{align*}
\widehat W^{(2)}_{MR} =&  \bigg [ \frac{1}{ \delta(L,\widehat \beta) \widehat f(Z|L) }\bigg \{Y -A  \Delta(L, \widehat \theta)\\& -\widehat E[Y|Z=-1,L] + \Delta(L, \widehat \theta) \widehat E[A|Z=-1,L] \bigg \}+ Z\Delta(L, \widehat \theta) \bigg ],
\end{align*}
respectively, where  $\delta(L,\widehat \beta)$ and $\Delta(L, \widehat \theta)$ are doubly robust estimators of $\delta(L)$ and $\Delta(L)$, and other nuisance estimators are estimated by maximum likelihood estimation.
It is also possible to use modern machine learning methods to estimate these nuisance parameters.
Under standard regularity conditions \citep{10.2307/1912526},
$\delta(L,\widehat \beta)$, $\Delta(L, \widehat \theta)$, $\widehat f(Z|L)$,  $\widehat E(Y|L,Z=-1)$, $\widehat E(A|L,Z=-1)$ converge in probability to $\delta(L,\beta^*)$, $\Delta(L, \theta^*)$, $f^*(Z|L)$,  $E^*(Y|L,Z=-1)$, $E^*(A|L,Z=-1)$.
 In Theorem~\ref{thm:triply2}, we show that
\begin{align*}
\Delta_{\cD}^{*(1)} = E \Big[  W^{*(1)}_{MR} I\{A=\cD(L)\}\Big] \quad \text{and} \quad  \Delta_{\cD}^{*(2)} = E\Big[  W^{*(2)}_{MR} I\{Z=\cD(L)\}\Big]
\end{align*}
are multiply robust in the sense of maximizing the value function  (or minimizing the weighted classification error) in the union model of the following models:

\vspace{0.2cm}

\noindent $\cM'_1$: models for $f(Z|L) $ and $ \delta(L)$ are correct;

\vspace{0.4cm}

\noindent $\cM'_2$:  models for  $f(Z|L)$  and $ \Delta(L) $ are correct;

\vspace{0.4cm}

\noindent $\cM'_3$:  models for  $\Delta(L), E[Y|Z=-1,L], \delta(L)$, $ E[A|Z=-1,L]$ are correct;

\vspace{0.2cm}
\noindent where the weights are
\begin{align*}
W^{*(1)}_{MR} = &\bigg[ \frac{ZA}{\delta(L,\beta^*)f(Z|L,\nu^*)} \bigg\{Y- A\Delta(L,\theta^*)- E^*(Y|L,Z=-1) \\ &+ \Delta(L, \theta^*) E^*(A|L,Z=-1) \bigg \}  +A \Delta(L,\theta^*) \bigg],
\end{align*}
and
\begin{align*}
W^{*(2)}_{MR} = &\bigg[ \frac{1}{\delta(L,\beta^*)f(Z|L,\nu^*)} \bigg\{Y- A\Delta(L,\theta^*)- E^*(Y|L,Z=-1) \\ &+ \Delta(L, \theta^*) E^*(A|L,Z=-1) \bigg\}  +Z \Delta(L,\theta^*) \bigg].
\end{align*}

\begin{theorem}
\label{thm:triply2}
Under Assumptions~\ref{asm:unconfoundedness2}-\ref{asm:weak} and standard regularity conditions, we have that
\begin{align*}
\arg\max_{\cD} \Delta_{\cD}^{*(1)} =\arg\max_{\cD} \Delta_{\cD}^{*(2)}= \arg\max_{\cD} E[\Delta(L) \cD(L)],
\end{align*}
under the union model $\cM'_{union}=\cM'_1 \cup \cM'_2 \cup  \cM'_3$.
\end{theorem}
Consequently, following the theoretical results established in Section~\ref{sec:theory} in the Supplementary Material, the risk of the estimated treatment regime $\widehat \cD$ converges to the Bayes risk in probability.

\section{Simulation experiments}\label{sec:simulation}
\yifan{In this section, we report extensive simulation studies comparing the proposed estimators to outcome weighted learning \citep{zhao2012estimating} and residual weighted learning \citep{rwl}, which are in principle valid only under unconfounded treatment.}

\subsection{Simulation settings}

We generated $L$ from a uniform distribution on $[-1,1]^5$. Treatment $A$ was generated under a logistic regression with success probability,
\begin{align*}
\Pr(A=1|L,Z,U)=\expit\{2L^{(1)}+2.5Z-0.5U\},
\end{align*}
with $Z$ a Bernoulli event with probability 1/2, and $U$ from a bridge distribution with parameter $\phi=1/2$.
By a theorem  of \cite{10.1093/biomet/90.4.765}, the above data generating mechanism ensures that there exists a vector $\alpha$ such that $\logit\{\Pr(A=1|L,Z)\}= \alpha^T(1,L,Z)$, so that upon marginalizing over $U$ the model for $f(A|L,Z)$ remains a logistic regression.
\yifan{Additional simulations are conducted in the Supplementary Material (Tables~\ref{simu3}-\ref{simu6}) to illustrate how the strength of the instruments affects the variance of the estimated value functions.}

 The outcome $Y$ was generated differently in each scenario as described below. The sample size was 500 for each scenario. We repeated the simulation 500 times. A large independent test set with 10000 subjects was used to evaluate the performance of different methods. \yifan{Additional simulation results with sample sizes 250 and 1000 are shown in Tables~\ref{simu2501}-\ref{simu1002} in the Supplementary Material.}

The proposed methods were implemented according to Section~\ref{sec:2.3} with $\widehat \delta(L)=\widehat f(A|L,Z=1) - \widehat f(A|L,Z=-1)$ and $\widehat f(Z|L)$, where $\widehat f(A|L,Z)$ and $\widehat f(Z|L)$ were estimated from logistic regression models.
Outcome weighted learning and residual weighted learning likewise used $\widehat f(A|L,Z)$ for $f(A|L,Z)$.
\yifan{In addition, we implemented multiply robust weights with
correctly specified $M_1'$ but incorrect models for $M_2'$ and $M_3'$ to better understand the sensitivity of the proposed methods to misspecified nuisance models.
In particular, $E[Y|L,Z]$ was estimated by a linear regression, denoted by $\widehat E[Y|L,Z]$, and
$\Delta(L)$ was estimated by $\widehat \Delta(L) = \widehat \Delta^Y(L)/\widehat \delta(L)$, where $\widehat \Delta^Y(L) = \widehat E[Y|L,Z=1] - \widehat E[Y|L,Z=-1]$.
Furthermore, $\widehat E[A|Z=-1,L]$ and $\widehat \delta(L)$ were estimated by logistic regressions as specified above.}
 Both linear and Gaussian kernels
were considered for all methods. We applied cross-validation for choosing tuning parameters by searching over a pre-specified finite set following \cite{zhao2012estimating}.

We considered the following four scenarios for the outcome model, with both linear and nonlinear outcome models:
\begin{align*}
(1) \quad \quad  & Y = h(L) + q(L) A + \epsilon,  \\
(2) \quad \quad  & Y = h(L) + q(L)A + 0.5U + \epsilon, \\
(3) \quad \quad  & Y = \exp\{h(L) + q(L) A\}+ \epsilon,  \\
(4) \quad \quad  & Y = \exp\{h(L) + q(L)A\} + U + \epsilon,
\end{align*}
where the error term $\epsilon$ followed the standard normal distribution, and
\begin{align*}
h(L) &= (0.5+0.5L^{(1)} +0.8L^{(2)} +0.3 L^{(3)} -0.5L^{(4)} +0.7 L^{(5)}),\\
q(L)&=(0.2- 0.6 {L^{(1)}} - 0.8 {L^{(2)}}).
\end{align*}
Scenarios 1 and 3 were considered in \cite{zhou2017augmented}. Scenarios 2 and 4 are modifications of 1 and 3 by adding unmeasured confounding.

\subsection{Numerical results}

Table~\ref{simu1} reports the mean and standard deviation of value functions evaluated at estimated optimal regimes in test samples.
Table~\ref{simu2} reports the mean and standard deviation of correct classification rates in test samples.
\yifan{It is interesting to note that, by leveraging Equation~\eqref{eq:equiv}, the proposed estimators with $\widehat W^{(1)}$ and $\widehat W^{(2)}$ give the same estimated treatment regime, so do the multiply robust estimators with $\widehat W^{(1)}_{MR}$ and $\widehat W^{(2)}_{MR}$.}

 \yifan{In Scenario~1,
as unconfoundedness assumption~\ref{asm:unconfoundedness} holds, it is not surprising that all methods perform similarly.
In Scenario 2, where $U$ is present, the treatment assignment and the outcome are confounded. Outcome weighted learning and residual weighted learning in this case fail to find an optimal regime.  Our estimated treatment regime performs much better for both linear and Gaussian kernels.
In Scenarios 3 and 4, we again observe a \mbox{consistent} pattern that the proposed method performs much better in the presence of unmeasured confounding.
Furthermore, in almost all scenarios, the multiply robust estimator improves upon inverse weighted estimator and residual weighted learning improves upon outcome weighted learning.}


\begin{table}[!h]
\begin{center}
\caption{\label{simu1}
Simulation results: Mean $\times 10^{-2} $ (sd $\times 10^{-2}$) of value functions}
\begin{tabular}{cccccc}
\noalign{\smallskip}
\noalign{\smallskip}
 & Kernel   & OWL & RWL & IV-IW   & IV-MR  \\
\noalign{\smallskip}
\multirow{2}{*}{1} & Linear  &  96.0~(2.9)   &  97.3~(1.8)  &  95.6~(4.8)  &  96.9~(2.9)    \\
                & Gaussian &   87.6~(7.5)  &  95.4~(3.1)  &  89.3~(9.0)  &  93.4~(5.7)    \\
\noalign{\smallskip}
\multirow{2}{*}{2} & Linear & 35.9~(17.3)  &  37.9~(18.4) & 91.5~(7.6) &  92.3~(7.2)      \\
                & Gaussian  &  61.2~(9.9)   & 61.4~(10.2)  & 81.8~(11.3) & 85.4~(9.7)    \\
\noalign{\smallskip}
\noalign{\smallskip}
\multirow{2}{*}{3} & Linear &  356.5~(4.4) &  359.4~(2.5) & 358.5~(3.4)  &  358.9~(3.1)   \\
                & Gaussian   &  297.0~(33.0) &  356.6~(4.4) & 315.8~(34.6)  &  354.7~(8.4)    \\
\noalign{\smallskip}
\multirow{2}{*}{4} & Linear &  275.1~(4.6)  & 275.8~ (6.6)  & 349.1~(12.1) &  349.8~(10.1)   \\
                & Gaussian   &  280.4~(13.2)  & 298.2~(14.0) & 308.3~(33.8) & 331.2~(23.0)    \\
\end{tabular}
\end{center}
\source{OWL:  outcome weighted learning; RWL: residual weighted learning; IV-IW: the proposed estimator with weight $\widehat W^{(1)}$ or $\widehat W^{(2)}$; IV-MR: the proposed multiply robust estimator with weight $\widehat W^{(1)}_{MR} $ or $\widehat W^{(2)}_{MR} $. The empirical optimal value functions are 0.998, 0.995, 3.636, 3.630 for four scenarios, respectively.}
\end{table}

\begin{table}[!h]
\begin{center}
\caption{\label{simu2}
Simulation results: Mean $\times 10^{-2}$ (sd $\times 10^{-2}$) of correct classification rates}
\begin{tabular}{cccccc}
\noalign{\smallskip}
\noalign{\smallskip}
 & Kernel  & OWL & RWL & IV-IW & IV-MR  \\
\noalign{\smallskip}
\multirow{2}{*}{1}  & Linear &  88.0~(4.2)   &  90.2~(3.3) &  87.8~(5.6)  &  89.7~(4.0)      \\
                & Gaussian  &   79.3~(7.5)  &  87.5~(4.1) & 81.4~(8.5)  &  85.4~(6.0)    \\
\noalign{\smallskip}
\multirow{2}{*}{2} & Linear  &  42.0~(10.0) &   42.9~(10.3) & 83.0~(7.9)  &  84.0~(7.6)    \\
                & Gaussian &  57.1~(6.6)  &  57.5~(6.9)  & 74.5~(9.5)  &  77.5~(8.6)    \\
\noalign{\smallskip}
\noalign{\smallskip}
\multirow{2}{*}{3} & Linear &  88.7~(3.6)  &  90.0~(3.7) & 90.3~(3.4)  &  89.9~(3.2)   \\
                & Gaussian  &  71.1~(9.8)  &  87.9~(4.6) & 75.7~(11.7)  &  88.0~(5.2)  \\
\noalign{\smallskip}
\multirow{2}{*}{4} & Linear &  37.3~(2.0)  &  37.5~(2.7) & 83.1~(7.4) &  83.2~(6.7)   \\
                & Gaussian  &  44.8~(6.5)  &  48.4~(7.7)  & 69.0~(11.3) &  74.8~(10.2) \\
\end{tabular}
\end{center}
\source{OWL:  outcome weighted learning; RWL: residual weighted learning; IV-IW: the proposed estimator with weight $\widehat W^{(1)}$ or $\widehat W^{(2)}$; IV-MR: the proposed multiply robust estimator with weight $\widehat W^{(1)}_{MR} $ or $\widehat W^{(2)}_{MR}$.}
\end{table}

\yifan{As can be seen from Tables~\ref{simu3}-\ref{simu6} in the Supplementary Material, a higher compliance rate generally leads to a lower variance of the estimated regime in terms of both value functions and correct classification rates. In Tables~\ref{simu2501}-\ref{simu1002}, the observed patterns are consistent across different sample sizes, and as sample size increases, the proposed methods have higher prediction accuracy.}

\section{Data analysis}\label{sec:realdata}

In this section, we follow  \cite{angrist1998} and study a sample of married mothers with two or more children from 1980 census data. The data entail a publicly available sample from the U.S. 1980 census of married and unmarried mothers.
\cite{angrist1998} estimated the local average treatment effect of having a third child among
mothers with at least two children. \cite{athey2019} identified a conditional local average treatment effect given several covariates.

We seek to provide a personalized recommendation on whether or not a woman should plan to have a third child without compromising her ability to participate in the labor market. Therefore, we wish to discover optimal regimes for deciding to have three or more children in order to maximize the probability of remaining in the labor market. We included the following five covariates considered in \cite{athey2019}: the mother's age at the birth of her first child, her age at census time, her years of education and her race, as well as the father's income.
The outcome $Y$ was whether or not the mother worked in the year preceding the census. The treatment $A$ denoted whether the mother had three or more children at census time, and the instrument $Z$ was whether or not the mother's first two children were of the same sex.

In order to draw comparison between the various methods, we randomly selected 500 subjects as training set and 5000 subjects as test set from the original dataset including 561,459 subjects. This procedure was repeated 100 times.  We performed the analysis on the training dataset, and obtained the estimated optimal treatment regimes from the four methods evaluated in the previous section. The nuisance parameters were estimated as described in Section~\ref{sec:simulation} except for $E[Y|L,Z]$ which we modeled as a logistic regression because the outcome was binary. Tuning parameters were selected in the same way as Section~\ref{sec:simulation}.
Empirical values of estimated treatment regime $\widehat \cD$ were evaluated with
\begin{align*}
V = \frac{1}{n}\sum_{i=1}^n \widehat W^{(1)}_i I\{A_i = \widehat \cD(L_i)\},
\end{align*}
where $\widehat W_i^{(1)}$  was estimated by logistic regression according to Section~\ref{sec:2.3} using test dataset. A larger empirical value may be interpreted as a better performance.

Results are presented in Table~\ref{real1}.
Both proposed methods have higher values for linear and Gaussian kernels, and
Monte Carlo standard errors are comparable across all methods.
In addition, compared to outcome weighted learning, residual weighted learning has lower mean and higher variance of value functions for both linear and gaussian kernels, which is less likely to happen if no unmeasured confounding assumption holds.
A possible reason is that unmeasured confounding causes inaccurate estimation of outcome weighted learning and residual weighted learning. We selected 500 subjects and investigated the estimated linear decision rules. Intuitively, one might expect that having a third child would generally reduce a mother's labor participation even if the effects are heterogeneous. However, half of the coefficients including intercept of the estimated decision function for residual weighted learning appear to be positive.
Thus, the corresponding decision rule might be incorrectly recommending women to have a third child which may in fact reduce their labor participation. In contrast, most of the coefficients including intercept of the estimated decision function for the multiply robust estimator are negative and therefore, the corresponding decision rule seems to recommend the expected optimal policy.

\begin{table}[h]
\begin{center}
\caption{\label{real1} Real data application: Mean $\times 10^{-2}$ (sd $\times 10^{-2}$) of $V$}
\begin{tabular}{ccccc}
\noalign{\smallskip}
\noalign{\smallskip}
Kernel  &  OWL & RWL & IV-IW & IV-MR  \\
\noalign{\smallskip}
Linear      & 60.5~(21.6) & 60.8~(25.1) &  61.2~(24.3) & 63.7~(24.1)  \\
Gaussian    & 64.5~(24.3) & 63.8~(25.1) &  64.6~(22.3) & 65.4~(23.7) \\
\end{tabular}
\end{center}\source{OWL:  outcome weighted learning; RWL: residual weighted learning; IV-IW: the proposed estimator with weight $\widehat W^{(1)}$ or $\widehat W^{(2)}$; IV-MR: the proposed multiply robust estimator with weight $\widehat W^{(1)}_{MR} $ or $\widehat W^{(2)}_{MR} $.}
\end{table}

\section{Complier optimal treatment regimes}\label{sec:cace}

In this section, we target complier optimal treatment regimes, i.e., treatment regimes that would optimize the potential outcome among compliers:
\begin{align*}
\cD^\dag(L) = \sign\{ E[ Y_1 - Y_{-1}| A_1 >A_{-1}, L ] \},
\end{align*}
where $A_z$ denotes the potential treatment under an intervention that sets the IV $Z$ to $z$.
We define compliers' value function,
\begin{align*}
\cV_c(\cD) = E\bigg[ I\{\cD(L)=1\} E[ Y_1 | A_1 >A_{-1}, L ] + I\{\cD(L)=-1\} E[ Y_{-1} | A_1 >A_{-1}, L ]   \bigg| A_1 > A_{-1}  \bigg],
\end{align*}
provided that $\Pr(A_1 > A_{-1})>0$.
Thus, complier optimal treatment regimes can be formulated as
\begin{align*}
\cD^\dag=\arg\max_{\cD} \cV_c(\cD) .
\end{align*}
In order to identify complier optimal treatment regimes, we make the following well-known assumption.
\begin{assumption} \emph{(Monotonicity)}
$\Pr(A_{1} \geq A_{-1}) = 1$.
\label{asm:monotonicity}
\end{assumption}

Assumption~\ref{asm:monotonicity} essentially rules out the existence of defiers in the population, i.e., with $A_{1}<A_{-1}$. Furthermore, we assume $Z$ to be a causal IV \citep{hernan2006epi}, i.e., the causal effects of the IV on the treatment and outcome are unconfounded given $L$ in the following sense.
\begin{assumption} \emph{(Causal IV)}
$Z \indep (A_z,Y_{z,a})|L$ for $z,a=\pm 1$.
\label{asm:causal IV}
\end{assumption}
Then we have the following two identification results analogous to Theorems~\ref{thm:identification} and \ref{thm:identification2}, respectively.
\begin{theorem}
Under Assumptions~\ref{IV Relevance}-\ref{Exclusion Restriction}, \ref{asm:IV positivity}, and \ref{asm:monotonicity}-\ref{asm:causal IV}, the compliers' value function is nonparametrically identified:
\begin{align*}
\cV_c(\cD)   = E \left[\frac{ZAYI\{A=\cD(L)\}}{\{\Pr(A=1|Z=1)-\Pr(A=1|Z=-1)\} f(Z|L) } \right].
\end{align*}
Therefore, complier optimal treatment regimes are given by
$$\arg\max_{\cD} E \left[\frac{ZAYI\{A=\cD(L)\}}{f(Z|L) } \right]. $$
\label{thm:identification3}
\end{theorem}

\begin{theorem}
Under Assumptions~\ref{IV Relevance}-\ref{Exclusion Restriction}, \ref{asm:IV positivity}, and \ref{asm:monotonicity}-\ref{asm:causal IV},
complier optimal treatment regimes are nonparametrically identified,
\begin{align}
\arg\max_{\cD} \cV_c(\cD) = \arg\max_{\cD} E \left[\frac{YI\{Z=\cD(L)\}}{ f(Z|L) } \right].
\end{align}
\label{thm:identification4}
\end{theorem}

\begin{remark}
Theorem~\ref{thm:identification4} is somewhat surprising, as it suggests that one can in fact identify optimal treatment regimes for those who would always comply to their assigned treatment without observing their realized treatment values.  Intuition about this result is gained upon noting that under monotonicity, the effect of $Z$ on $Y$ within levels of $L$ is proportional to the causal \mbox{effect} of $A$ on $Y$ among compliers within levels of $L$, where the proportionality constant within levels of $L$ is the (nonnegative) additive causal effect of $Z$ on $A$.
Consequently, under monotonicity assumption, optimizing the value function $E[E(Y|Z=\cD(L),L)]$ with respect to the treatment assignment policy \yifan{(i.e., applying standard outcome weighted learning w.r.t. $Z$)},
is equivalent to optimizing the value function among compliers with respect to the treatment $\arg\max E(Y_{\cD(L)}|A_1>A_{-1})$, a task which can therefore be accomplished without directly observing the treatment variable. Thus, it is possible to learn who might benefit from the intervention even when one does not observe $A$ and therefore cannot identify the complier average treatment effect.
\end{remark}

\begin{remark}\label{remark3}
\yifan{
Typically, the first step in IV analyses is to assess the strength of the instrument by calculating the compliance rate. The strength of an IV is directly related to the performance of the corresponding estimator.
When one does not observe the treatment, one can not guarantee that the IV and the treatment are strongly
associated, in which case weak IV problem cannot necessarily be assessed  \citep{bound1995,small2008war,baiocchi2010building,baiocchi2014instrumental,ertefaie2018}. In well designed randomized experiments subject to non-compliance, although not perfect, compliance nevertheless typically remains relatively high (i.e., $\geq$ 80\%).
}
\end{remark}

Empirical versions of equations in Theorems~\ref{thm:identification3} and \ref{thm:identification4} give rise to estimators of value function and optimal treatment regime, respectively. Furthermore, it is relatively straightforward to show that results analogous to Section~\ref{sec:theory}
in the Supplementary Material also hold for complier optimal treatment regimes.

\section{Discussion}\label{sec:discussion}
In this paper,
we have proposed a general instrumental variable approach to learning optimal treatment regimes under endogeneity. To our knowledge, this is the first result for estimating optimal regimes when no unmeasured confounding fails to hold.
Specifically, we established identification of both value function $E[Y_{\cD(L)}]$ for a given regime $\cD$ and optimal regimes $\arg \max_{\cD} E[Y_{\cD(L)}]$  with the aid of a binary IV.
We also constructed novel multiply robust classification-based estimators.
Furthermore, we proposed to identify and estimate optimal treatment regimes among compliers under monotonicity. In the latter case,  we established the somewhat surprising result that complier optimal treatment regimes can be consistently estimated without accessing compliance information. Our approach was illustrated via extensive simulation studies and a real data application.

\yifan{
The proposed methods may be improved or extended in several \mbox{directions}.
Sometimes the values of instruments are unknown and must be estimated using the data \citep{baiocchi2014instrumental,ertefaie2018}, e.g., preference-based IVs \citep{brookhart2007preference}.
Understanding the implication for inference of empirically defining IV is a fruitful avenue of future research.
Moreover, as mentioned in Remark~\ref{remark3}, it is known that weak IVs can be problematic \citep{bound1995,small2008war,baiocchi2010building,baiocchi2014instrumental,ertefaie2018}. It may be possible to estimate optimal treatment regimes by empirically building stronger instruments \citep{baiocchi2010building,zubizarreta2013stronger,baiocchi2014instrumental,ertefaie2018}.
}

The proposed methods can also be modified in case of a censored survival outcome by accounting for possibly dependent censoring, thus providing extensions to \cite{zhao2015doubly,cui2017tree} to leverage an IV.
In addition, trials with multiple treatment arms occur frequently. Thus a potential extension of our method is in the direction of multicategory classification \citep{JMLR:v18:17-003,Zhou2018OutcomeWeightedLF,Qi2019}. Furthermore, personalized dose finding \citep{chen2016dosefinding,zhou2018dosedr} with unmeasured confounding is also of interest. Finally, it would be of interest to extend our methods to mobile health dynamic treatment regimes where a sequence of decision rules need to be learned under endogeneity \citep{Robins2004,Zhang2013RobustEO,laber2014,zhao2015new,luckett2019}.

\section{Funding}
The authors were supported by NIH funding: R01CA222147 and R01AI127271.

\newpage
\appendix
\begin{center}
{\LARGE Appendix}
\end{center}
\section{Proof of Theorem~\ref{thm:identification}}
\begin{proof} We first note that
\begin{align*}
& E \left [\frac{  ZI\{\cD(L)=A\} YA }{ \delta(L)f(Z|L) } \right]\\
= & E \left[ \sum_a \frac{ Z I\{A=a\} I\{\cD(L)=a\}Y_a a}{ \delta(L)f(Z|L)}  \right] \\
= & E \left[ \sum_a \frac{Z I\{A=a\} I\{\cD(L)=a\}E[Y_a|L,U] a}{ \delta(L)f(Z|L)}  \right] \\
= & E \left[ \sum_a \frac{Z \Pr(A=a| L,U,Z) I\{\cD(L)=a\}E[Y_a|L,U] a}{ \delta(L)f(Z|L)}  \right] \\
= & E \left[ \frac{ \Pr(A=1| L,U,Z=1) I\{\cD(L)=1\}E[Y_1|L,U]}{ \delta(L)}  \right] \\
- & E \left[ \frac{ \Pr(A=1| L,U,Z=-1) I\{\cD(L)=1\}E[Y_1|L,U]}{ \delta(L)}  \right] \\
- & E \left[ \frac{ \Pr(A=-1| L,U,Z=1) I\{\cD(L)=-1\}E[Y_{-1}|L,U] }{ \delta(L)}  \right]\\
+ & E \left[ \frac{\Pr(A=-1| L,U,Z=-1) I\{\cD(L)=-1\}E[Y_{-1}|L,U]}{ \delta(L)}  \right] \\
= & E \left[ \frac{  \left[\Pr(A=1| L,U,Z=1)-\Pr(A=1| L,U,Z=-1) \right]I\{\cD(L)=1\}E[Y_1|L,U]  }{\delta(L)}  \right] \\
+ & E \left[ \frac{  \left[\Pr(A=1| L,U,Z=1)-\Pr(A=1| L,U,Z=-1) \right]I\{\cD(L)=-1\}E[Y_{-1}|L,U]  }{\delta(L)}  \right] \\
\equiv & (I).
\end{align*}

  In order to maximize counterfactual mean $E[Y_{\cD(L)}]$, we only need Assumption~\ref{asm:weak} rather than Assumption~\ref{asm:strong}. To see this, note that
\begin{align*}
& E[Y_1|L,U]I\{\cD(L)=1\}+E[Y_{-1}|L,U]I\{\cD(L)= -1\}\\
= & (E[Y_1|L,U]-E[Y_{-1}|L,U])I\{\cD(L)=1\}+E[Y_{-1}|L,U]I\{\cD(L)=1\} + E[Y_{-1}|L,U]I\{\cD(L)=-1\}\\
= & (E[Y_1|L,U]-E[Y_{-1}|L,U])I\{\cD(L)=1\}+E[Y_{-1}|L,U].\\
\end{align*}

By Assumption~\ref{asm:weak}, we further have that
\begin{align*}
(I) = & E\left[  \frac{ \Pr(A=1|U,L,Z=1)- \Pr(A=1|U,L,Z=-1)  }{ \delta(L) \big\{ (E[Y_1|L,U]-E[Y_{-1}|L,U])I\{\cD(L)=1\}+ E[Y_{-1}|L,U] \big\}^{-1} } \right] \\
= & E \bigg \{ \left( E[Y_1|L,U]-E[Y_{-1}|L,U] \right)I\{\cD(L)=1\}\\
& +  \frac{\Pr(A=1|U,L,Z=1)- \Pr(A=1|U,L,Z=-1) }{\delta(L)} E[Y_{-1}|L,U] \bigg \}\\
= & E \big [  E(Y_1-Y_{-1}|L) I\{\cD(L)=1\}\big ] \\
 & + E\left \{ \frac{[\Pr(A=1|U,L,Z=1)- \Pr(A=1|U,L,Z=-1)]E[Y_{-1}|L,U] }{\delta(L)}  \right \} \\
= & E \big [  E(Y_1-Y_{-1}|L) I\{\cD(L)=1\}\big ] + E[\kappa(L,U)],
\end{align*}
where the second term $E[\kappa(L,U)]$ doesn't depend on $\cD$. Recall that
\begin{align*}
\arg\max_\cD E[Y_{\cD(L)}]=\arg\max_{\cD} E[E(Y_1-Y_{-1}|L)I\{\cD(L)=1\}],
\end{align*}
so maximizing $(I)$ is equivalent to maximizing $E[Y_{\cD(L)}]$.

Furthermore, by Assumption~\ref{asm:strong},
\begin{align*}
(I)= & E \left(  I\{\cD(L)=1\}E[Y_1|L,U]+I\{\cD(L)=-1\}E[Y_{-1}|L,U]  \right)  \\
= & E[Y_{\cD(L)}].
\end{align*}
\end{proof}

\section{Proof of Theorem~\ref{thm:identification2}}
\begin{proof}
\begin{align*}
& E \left [\frac{  I\{\cD(L)=Z\} Y }{ \delta(L)f(Z|L) } \right]\\
= & E \left[ \sum_a \frac{I\{\cD(L)=Z\}Y_a I\{A=a\}}{ \delta(L)f(Z|L)}  \right] \\
= & E \left[ \sum_a \frac{ I\{\cD(L)=Z\}E[Y_a|L,U] \Pr(A=a|L,U,Z)}{ \delta(L)f(Z|L)}  \right] \\
= & E \left[ \frac{ \Pr(A=1| L,U,Z=1) I\{\cD(L)=1\}E[Y_1|L,U]}{ \delta(L)}  \right] \\
+ & E \left[ \frac{ \Pr(A=1| L,U,Z=-1) I\{\cD(L)=-1\}E[Y_1|L,U]}{ \delta(L)}  \right] \\
+ & E \left[ \frac{ \Pr(A=-1| L,U,Z=1) I\{\cD(L)=1\}E[Y_{-1}|L,U] }{ \delta(L)}  \right]\\
+ & E \left[ \frac{\Pr(A=-1| L,U,Z=-1) I\{\cD(L)=-1\}E[Y_{-1}|L,U]}{ \delta(L)}  \right] \\
= & E \left[ \frac{  \left[\Pr(A=1| L,U,Z=1)-\Pr(A=1| L,U,Z=-1) \right]I\{\cD(L)=1\}E[Y_1|L,U]  }{\delta(L)}  \right] \\
+ & E \left[ \frac{  \left[\Pr(A=1| L,U,Z=1)-\Pr(A=1| L,U,Z=-1) \right]I\{\cD(L)=-1\}E[Y_{-1}|L,U]  }{\delta(L)}  \right] \\
+ & E \left[ \frac{ \Pr(A=1| L,U,Z=-1) E[Y_1|L,U]}{ \delta(L)}  \right] \\
+ & E \left[ \frac{ \Pr(A=-1| L,U,Z=-1) E[Y_{-1}|L,U]}{ \delta(L)}  \right]\\
= & E \big (  E[Y_1-Y_{-1}|L] I\{\cD(L)=1\}\big ) + E[\kappa(L,U)],
\end{align*}
where
\begin{align*}
\kappa(L,U)= &  \frac{[\Pr(A=1|U,L,Z=1)- \Pr(A=1|U,L,Z=-1)]E[Y_{-1}|L,U] }{\delta(L)}  \\
+ &  \frac{ \Pr(A=1| L,U,Z=-1) E[Y_1|L,U]+ \Pr(A=-1| L,U,Z=-1) E[Y_{-1}|L,U]}{ \delta(L)}.
\end{align*}

\end{proof}

\section{Proof of Theorem~\ref{thm:triply2}}
\begin{proof}
For $\Delta_{\cD}^{*(1)}$, we have that
\begin{align*}
& 2E\Big [W^{*(1)}_{MR} I\{A=\cD(L)\}\Big ]\\
= & E\Big [ W^{*(1)}_{MR}  [2I\{A= \cD(L)\}-1] \Big ] +  E \Big [ W^{*(1)}_{MR} \Big ]\\
= & E\Big [W^{*(1)}_{MR} A\cD(L) \Big] +   E \Big [W^{*(1)}_{MR} \Big] \\
= & E\bigg [\bigg \{ \frac{Z}{\delta(L,\beta^*)f^*(Z|L) }\bigg [Y -A \Delta(L,\theta^*) -E^*[Y|Z=-1,L] +\\
&\Delta(L,\theta^*)E^*[A|Z=-1,L] \bigg ]+ \Delta(L,\theta^*) \bigg \} \cD(L) \bigg ] +  E \Big[ W^{*(1)}_{MR} \Big]\\
= & E\Big [ \Delta(L) \cD(L) \Big ] +  E \Big[ W^{*(1)}_{MR} \Big],
\end{align*}
where the last equality holds under any of $\cM'_1$, $\cM'_2$, or $\cM'_3$,  and the proof follows a similar argument of Theorem 6 in \cite{wang2018bounded}.

For $\Delta_{\cD}^{*(2)}$, we have that
\begin{align*}
& 2 E\Big [W^{*(2)}_{MR} I\{Z=\cD(L)\}\Big ]\\
= & E\Big [ W^{*(2)}_{MR}  [2I\{Z= \cD(L)\}-1] \Big ] +  E \Big [ W^{*(2)}_{MR} \Big ]\\
= & E\Big [W^{*(2)}_{MR} Z\cD(L) \Big] +   E \Big [W^{*(2)}_{MR} \Big] \\
= & E\bigg [\bigg \{ \frac{Z}{\delta(L,\beta^*)f^*(Z|L) }\bigg [Y -A \Delta(L,\theta^*) -E^*[Y|Z=-1,L] +\\
&\Delta(L,\theta^*)E^*[A|Z=-1,L] \bigg ]+ \Delta(L,\theta^*) \bigg \} \cD(L) \bigg ] +  E \Big[  W^{*(2)}_{MR} \Big]\\
= & E\Big [ \Delta(L) \cD(L) \Big ] +  E \Big[ W^{*(2)}_{MR} \Big],
\end{align*}
 where the last equality holds under any of $\cM'_1$, $\cM'_2$, or $\cM'_3$. This completes our proof.
\end{proof}

\section{Proof of Theorem~\ref{thm:identification3}}
\begin{proof}  We have the following equality
\begin{align*}
& E \left [\frac{  I\{\cD(L)=A\}YAZ}{f(Z|L) } \right]\\
= & E \left[ \sum_a \frac{I\{\cD(L)=A\}Y_a I\{A=a\}aZ}{ f(Z|L)}  \right] \\
= & E \left[ \sum_z \sum_a \frac{ I\{\cD(L)=a\} Y_a I\{A_z=a\} I\{Z=z\}az}{ f(Z|L)}  \right] \\
= & E \left[ \sum_z \sum_a I\{\cD(L)=a\} Y_a az I\{A_z=a\} \right] \\
= & E \left[ I\{A_1=1\} I\{\cD(L)=1\}Y_1  \right] - E \left[ I\{A_{-1}=1\} I\{\cD(L)=1\}Y_1  \right] \\
- & E \left[ I\{A_1=-1\} I\{\cD(L)=-1\}Y_{-1}  \right] + E \left[ I\{A_{-1}=-1\} I\{\cD(L)=-1\}Y_{-1}  \right] \\
= & E \left[ I\{A_1=1\} I\{\cD(L)=1\}Y_1  \right] - E \left[ I\{A_{-1}=1\} I\{\cD(L)=1\}Y_1  \right] \\
+ & E \left[ I\{A_1=-1\} I\{\cD(L)=1\}Y_{-1}  \right] - E \left[ I\{A_{-1}=-1\} I\{\cD(L)=1\}Y_{-1}  \right] \\
- & E \left[ I\{A_1=-1\} Y_{-1}  \right] + E \left[ I\{A_{-1}=-1\} Y_{-1}  \right] \\
= & E \bigg[ I\{\cD(L)=1\} Y_1[ I\{A_1=1\} - I\{A_{-1}=1\}]  \bigg] \\
+ & E \bigg[ I\{\cD(L)=1\} Y_{-1}[ I\{A_1=-1\} - I\{A_{-1}=-1\}]  \bigg] + \kappa \\
= & E \bigg[ I\{\cD(L)=1\} E[ Y_1 - Y_{-1}| A_1 >A_{-1}, L ]   \big[ I\{A_1=1\} - I\{A_{-1}=1\} \big]  \bigg] + \kappa \\
= & E \bigg[ I\{\cD(L)=1\} E[ Y_1 - Y_{-1}| A_1 >A_{-1}, L ]   \bigg| A_1 > A_{-1}  \bigg] \Pr(A_1 > A_{-1}) + \kappa, \\
\end{align*}
\vspace{-1.5cm}

where
\vspace{-0.7cm}

\begin{align*}
\kappa= - & E \left[ I\{A_1=-1\} Y_{-1}  \right] + E \left[ I\{A_{-1}=-1\} Y_{-1}  \right]\\
=  & E \left[ [I\{A_1=1\}  - I\{A_{-1}=1\}] Y_{-1}  \right]\\
= & E[Y_{-1} | A_1 >A_{-1} ] \Pr( A_1 >A_{-1}).
\end{align*}
Thus, $$E \left[\frac{ZAYI\{\cD(L)=A\}}{\{\Pr(A=1|Z=1)-\Pr(A=1|Z=-1)\} f(Z|L) } \right]$$ identifies compliers' value function $\cV_c(\cD)$, i.e.,
\begin{align*}
E\bigg[ I\{\cD(L)=1\} E[ Y_1 | A_1 >A_{-1}, L ] + I\{\cD(L)=-1\} E[ Y_{-1} | A_1 >A_{-1}, L ]   \bigg| A_1 > A_{-1}  \bigg].
\end{align*}
\end{proof}

\section{Proof of Theorem~\ref{thm:identification4}}
\begin{proof}We have the following equality
\begin{align*}
& E \left [\frac{  I\{\cD(L)=Z\} Y }{f(Z|L) } \right]\\
= & E \left[ \sum_a \frac{I\{\cD(L)=Z\}Y_a I\{A=a\}}{ f(Z|L)}  \right] \\
= & E \left[ \sum_z \sum_a \frac{ I\{\cD(L)=z\} Y_a I\{A_z=a\} I\{Z=z\}}{ f(Z|L)}  \right] \\
= & E \left[ \sum_z \sum_a I\{\cD(L)=z\} Y_a I\{A_z=a\} \right] \\
= & E \left[ I\{A_1=1\} I\{\cD(L)=1\}Y_1  \right] + E \left[ I\{A_{-1}=1\} I\{\cD(L)=-1\}Y_1  \right] \\
+ & E \left[ I\{A_1=-1\} I\{\cD(L)=1\}Y_{-1}  \right] + E \left[ I\{A_{-1}=-1\} I\{\cD(L)=-1\} Y_{-1}  \right] \\
= & E \left[ I\{A_1=1\} I\{\cD(L)=1\}Y_1  \right] - E \left[ I\{A_{-1}=1\} I\{\cD(L)=1\}Y_1  \right] \\
+ & E \left[ I\{A_1=-1\} I\{\cD(L)=1\}Y_{-1}  \right] - E \left[ I\{A_{-1}=-1\} I\{\cD(L)=1\}Y_{-1}  \right] \\
+ & E \left[ I\{A_{-1}=1\} Y_1  \right] + E \left[ I\{A_{-1}=-1\} Y_{-1}  \right] \\
= & E \bigg[ I\{\cD(L)=1\} Y_1[ I\{A_1=1\} - I\{A_{-1}=1\}]  \bigg] \\
+ & E \bigg[ I\{\cD(L)=1\} Y_{-1}[ I\{A_1=-1\} - I\{A_{-1}=-1\}]  \bigg] + \kappa \\
= & E \bigg[ I\{\cD(L)=1\} E[ Y_1 - Y_{-1}| A_1 >A_{-1}, L ]   \big[ I\{A_1=1\} - I\{A_{-1}=1\} \big]  \bigg] + \kappa \\
= & E \bigg[ I\{\cD(L)=1\} E[ Y_1 - Y_{-1}| A_1 >A_{-1}, L ]   \big| A_1 > A_{-1}  \bigg] \Pr(A_1 > A_{-1}) + \kappa, \\
\end{align*}
where $\kappa= E \left[ I\{A_{-1}=1\} Y_1  \right] + E \left[ I\{A_{-1}=-1\} Y_{-1}  \right]$ does not depend on $\cD$. This completes our proof.
\end{proof}

\newpage
\begin{center}
{ \LARGE Supplementary Material}
\end{center}
\vspace{-0.5cm}

\section{Lower and upper bounds of $E\left[ Y_{\mathcal{D}(L)}\right]$}\label{bounds}

\begin{lemma}
Provided that $Z$ is a valid causal IV (as defined by \cite{Balke1997})  and outcome $Y$ is binary, we have the following lower and upper bounds of the value function,
\begin{align*}
&E \{\omega_1(L) [\mathcal{L}\left( L\right)
 I\left\{ \mathcal{D}(L)=1\right\} +\mathcal{L}_{-1}\left( L\right)]
 + \omega_{-1}(L) [ -\mathcal{U}\left( L\right) I\left\{ \mathcal{D}(L)=-1\right\} +\mathcal{L}_{1}\left( L\right)]\}\\
& \leq E[ \mathcal{L}_{1}(L) I\{\cD(L)=1\} +\mathcal{L}_{-1}(L) I\{\cD(L)=-1\}] \leq E\left[ Y_{\mathcal{D}
(L)}\right],\\
& E\{\omega_1(L) [\mathcal{U}\left( L\right)
 I\left\{ \mathcal{D}(L)=1\right\} +\mathcal{U}_{-1}\left( L\right)]
+ \omega_{-1}(L) [ -\mathcal{L}\left( L\right) I\left\{ \mathcal{D}(L)=-1\right\} +\mathcal{U}_{1}\left( L\right)]\}\\
& \geq E[ \mathcal{U}_{1}(L) I\{\cD(L)=1\} +\mathcal{U}_{-1}(L) I\{\cD(L)=-1\}] \geq E\left[ Y_{\mathcal{D}
(L)}\right],
\end{align*}
where $\omega_1(l),\omega_{-1}(l)\geq 0$, $\omega_1(l)+\omega_{-1}(l)=1$ for any $l$,
\begin{eqnarray*}
\mathcal{L}\left( l\right) = \mathcal{L}_{1}(l) - \mathcal{U}_{-1}(l) = \max
\left \{
  \begin{tabular}{c}
  $ p_{-1,-1,-1,l}+p_{1,1,1,l}-1$\\
  $ p_{-1,-1,1,l}+p_{1,1,1,l}-1$ \\
  $ p_{1,1,-1,l}+p_{-1,-1,1,l}-1$ \\
  $ p_{-1,-1,-1,l}+p_{1,1,-1,l}-1$ \\
  $2p_{-1,-1,-1,l}+p_{1,1,-1,l}+p_{1,-1,1,l}+p_{1,1,1,l}-2$ \\
  $p_{-1,-1,-1,l}+2p_{1,1,-1,l}+p_{-1,-1,1,l}+p_{-1,1,1,l}-2$ \\
  $p_{1,-1,-1,l}+p_{1,1,-1,l}+2p_{-1,-1,1,l}+p_{1,1,1,l}-2$ \\
  $p_{-1,-1,-1,l}+p_{-1,1,-1,l}+p_{-1,-1,1,l}+2p_{1,1,1,l}-2$ \\
  \end{tabular}
\right \}, \\
\end{eqnarray*}
\begin{eqnarray*}
\mathcal{U}\left( l\right) = \mathcal{U}_{1}(l) - \mathcal{L}_{-1}(l) =  \min
\left \{
  \begin{tabular}{c}
  $ 1- p_{1,-1,-1,l} - p_{-1,1,1,l}$\\
  $ 1- p_{-1,1,-1,l} - p_{1,-1,1,l}$ \\
  $ 1- p_{-1,1,-1,l} - p_{1,-1,-1,l}$ \\
  $ 1 - p_{-1,1,1,l} - p_{1,-1,1,l}$ \\
  $2- 2p_{-1,1,-1,l}- p_{1,-1,-1,l} - p_{1,-1,1,l} - p_{1,1,1,l}$ \\
  $2- p_{-1,1,-1,l}- 2p_{1,-1,-1,l} - p_{-1,-1,1,l} - p_{-1,1,1,l}$ \\
  $2- p_{1,-1,-1,l}- p_{1,1,-1,l} - 2p_{-1,1,1,l} - p_{1,-1,1,l}$ \\
  $2- p_{-1,-1,-1,l}- p_{-1,1,-1,l} - p_{-1,1,1,l} - 2p_{1,-1,1,l}$ \\
  \end{tabular}
\right \},
\end{eqnarray*}
\begin{eqnarray*}
\mathcal{L}_{-1}\left( l\right) =& \max
\left \{
  \begin{tabular}{c}
  $p_{1,-1,1,l}$ \\
  $p_{1,-1,-1,l}$ \\
  $p_{1,-1,-1,l} + p_{1,1,-1,l} - p_{-1,-1,1,l} - p_{1,1,1,l} $ \\
  $p_{-1,1,-1,l} + p_{1,-1,-1,l} - p_{-1,-1,1,l} - p_{-1,1,1,l} $ \\
  \end{tabular}
\right \}, \\
\mathcal{U}_{-1}\left( l \right) =& \min
\left \{
  \begin{tabular}{c}
  $1 - p_{-1,-1,1,l}$ \\
  $1- p_{-1,-1,-1,l}$ \\
  $p_{-1,1,-1,l} + p_{1,-1,-1,l} + p_{1,-1,1,l} + p_{1,1,1,l} $ \\
  $p_{1,-1,-1,l} + p_{1,1,-1,l} + p_{-1,1,1,l} + p_{1,-1,1,l} $ \\
  \end{tabular}
\right \}, \\
\mathcal{L}_{1}\left( l\right) =& \max
\left \{
  \begin{tabular}{c}
  $p_{1,1,-1,l}$ \\
  $p_{1,1,1,l}$ \\
  $-p_{-1,-1,-1,l} - p_{-1,1,-1,l} + p_{-1,-1,1,l} + p_{1,1,1,l} $ \\
  $-p_{-1,1,-1,l} - p_{1,-1,-1,l} + p_{1,-1,1,l} + p_{1,1,1,l} $ \\
  \end{tabular}
\right \}, \\
\mathcal{U}_{1}\left( l \right) =& \min
\left \{
  \begin{tabular}{c}
  $1 - p_{-1,1,1,l}$ \\
  $1- p_{-1,1,-1,l}$ \\
  $p_{-1,-1,-1,l} + p_{1,1,-1,l} + p_{1,-1,1,l} + p_{1,1,1,l} $ \\
  $p_{1,-1,-1,l} + p_{1,1,-1,l} + p_{-1,-1,1,l} + p_{1,1,1,l} $ \\
  \end{tabular}
\right \},
\end{eqnarray*}
and $p_{y,a,z,l}$ denotes $\Pr(Y=y,A=a|Z=z,L=l)$.

\label{lemma:bounds}
\end{lemma}

\begin{proof}
We consider construction of bounds for the value function with a valid IV.
Note that
\begin{align*}
E\left[ Y_{\mathcal{D}(L)}|L\right] &=E\left( Y_{1}|L\right) I\left\{
\mathcal{D}(L)=1\right\} +E\left( Y_{-1}|L\right)I\left\{
\mathcal{D}(L)=-1\right\},\\
E\left[ Y_{\mathcal{D}(L)}|L\right] &=E\left( Y_{1}-Y_{-1}|L\right) I\left\{
\mathcal{D}(L)=1\right\} +E\left( Y_{-1}|L\right),\\
E\left[ Y_{\mathcal{D}(L)}|L\right] &=E\left( Y_{-1}-Y_{1}|L\right) I\left\{
\mathcal{D}(L)=-1\right\} +E\left( Y_{1}|L\right).
\end{align*}
By the results from \cite{Balke1997}, one has the following bounds,
\begin{align}
&\omega_1(L) [\mathcal{L}\left( L\right)
 I\left\{ \mathcal{D}(L)=1\right\} +\mathcal{L}_{-1}\left( L\right)]
+ \omega_{-1}(L) [ -\mathcal{U}\left( L\right) I\left\{ \mathcal{D}(L)=-1\right\} +\mathcal{L}_{1}\left( L\right)] \nonumber \\
&\leq \mathcal{L}_{1}(L) I\{\cD(L)=1\} +\mathcal{L}_{-1}(L) I\{\cD(L)=-1\} \leq E\left[ Y_{\mathcal{D}
(L)}|L\right], \label{s:1}\\
 &\omega_1(L) [\mathcal{U}\left( L\right)
 I\left\{ \mathcal{D}(L)=1\right\} +\mathcal{U}_{-1}\left( L\right)]
+ \omega_{-1}(L) [ -\mathcal{L}\left( L\right) I\left\{ \mathcal{D}(L)=-1\right\} +\mathcal{U}_{1}\left( L\right)]\nonumber \\
&\geq  \mathcal{U}_{1}(L) I\{\cD(L)=1\} +\mathcal{U}_{-1}(L) I\{\cD(L)=-1\} \geq E\left[ Y_{\mathcal{D}
(L)}|L\right], \label{s:2}
\end{align}
where $\omega_1(L),\omega_{-1}(L)\geq 0$, $\omega_1(L) + \omega_{-1}(L)=1$,
$\mathcal{L}\left( L\right) $ and $\mathcal{U}\left( L\right) $ are
lower and upper bounds for $E\left( Y_{1}-Y_{-1}|L\right) $ given by \cite{Balke1997},
while $\mathcal{L}_{-1}\left( L\right) $, $\mathcal{U}_{-1}\left( L\right) $, $\mathcal{L}_{1}\left( L\right) $, $\mathcal{U}_{1}\left( L\right) $
are lower and upper bounds for $E\left( Y_{-1}|L\right) $ and $E\left( Y_{1}|L\right) $ obtained by \cite{Balke1997}.
Therefore, we complete the proof by taking expectations on both sides of Equations~\eqref{s:1} and \eqref{s:2}.
\end{proof}
Because it is not possible to directly maximize $E[ Y_{\mathcal{D}(L)}]$, one may nevertheless proceed by maximizing the minimum value function
$E[ \mathcal{L}_{1}(L) I\{\cD(L)=1\} +\mathcal{L}_{-1}(L) I\{\cD(L)=-1\}]$
and its lower bounds
 with user-specified weights $\omega_1(\cdot)$ and $\omega_{-1}(\cdot)$ which may reflect personal preferences. For instance, if $A = -1$ refers to placebo, the
safest regime might be maximizing $E[ \mathcal{L}\left( L\right) I\{ \mathcal{D}(L)=1\} ]$, i.e., assigning only $A = 1$ to those for whom $\mathcal{L}(L) > 0$.
Note that maximizing $E\left[ \mathcal{L}\left( L\right) I\left\{ \mathcal{D}(L)=1\right\} \right]$ and $E\left[ -\mathcal{U}\left( L\right) I\left\{ \mathcal{D}(L)=-1\right\} \right]$ would recommend two conflicting treatments to patients whose $(\mathcal{L}(L),\mathcal{U}(L))$ covers 0, i.e., the treatment decision remains ambiguous to those patients. We caution that the interval estimate for whom $(\mathcal{L}(L),\mathcal{U}(L))$ covers 0 might not be further narrowed down to an optimal treatment decision given the overwhelming uncertainty inside IV bounds.

\section{Fisher consistency, excess risk bound and universal consistency of the estimated regime\label{sec:theory}}

In this section, we establish Fisher consistency, excess risk bound and universal consistency of the estimated treatment regime. We focus on our first estimator, however, the results hold for the second estimator.

\subsection{Preliminaries}

Define the following risk
\begin{align*}
R(g) \equiv E [W I\{A\neq \sign (g(L))\} ],
\end{align*}
where $W=AZY/(\delta(L)f(Z|L))$. The optimal decision function associated with the optimal treatment regime $\cD^*$ is defined as $g^* \equiv \arg \min_{g\in \cG} R(g)$ and corresponding Bayes risk is $ R^*  \equiv R(g^*)$, where $\cG$ is the class of all measurable functions.

We also define the $\phi$-risk
\begin{align*}
R_\phi(g) \equiv E [|W| \phi(\sign(W)Ag(L)) ],
\end{align*}
where $\phi$ is the hinge loss function. The minimal $\phi$-risk $R_\phi^* \equiv \inf_{g \in \cG} R_\phi(g)$ and $g_\phi^* \equiv \arg\min_{g \in \cG} R_\phi(g)$.

\subsection{Fisher consistency and excess risk bound}
Note that Theorem 2.1 of \cite{zhou2017augmented} shows that
Fisher consistency holds if and only if $\phi'(0)$ exists and $\phi'(0)<0$ provided that the loss function $\phi$ is convex. The hinge loss function $\phi$ satisfies this condition which essentially implies the following Fisher consistency.

\begin{lemma}\label{fisher}  Under Assumptions \ref{asm:unconfoundedness2}-\ref{asm:weak}, $R^*=R(g^*_\phi)$.
\end{lemma}

The following theorem states that for any measurable decision function $g$, the excess risk under 0-1 loss is bounded by the excess $\phi$-risk.

\begin{lemma}
\label{thm:excess}
Under Assumptions \ref{asm:unconfoundedness2}-\ref{asm:weak}, for any measurable decision function $g$, we have that
\begin{align*}
R(g)-R^* \leq R_\phi(g)-R^*_\phi.
\end{align*}
\end{lemma}

The proof follows from Theorem 2.2 of \cite{zhou2017augmented}. Lemma~\ref{thm:excess} implies that the loss of the value
function due to the individualized treatment regime $\cD$ associated with the decision function $g$ can be bounded by the excess risk under the hinge loss. This excess bound also serves as an intermediate step for investigating the universal consistency of the estimated treatment regime.

\subsection{Consistency of the estimated treatment regime}

In this section, we establish the universal consistency of the estimated treatment regime with a universal kernel (e.g., Gaussian kernel). Estimation error has two potential sources. The first is from the approximation error associated with
$\cH_K$. The second is the uncertainty in estimated weights.

Before stating the universal consistency result of the estimated \mbox{treatment} regime, we first introduce the concept of universal kernels \citep{Steinwart:2008:SVM:1481236}. A continuous kernel $K$ on a compact metric space $\cL$ is called universal if its associated reproducing kernel Hilbert space (RKHS) $\cH_K$ is dense in $C(\cL)$, where $C(\cL)$ is the space of all continuous functions
on the compact metric space $\cL$ endowed with the usual supremum norm.

Let $K$ be a universal kernel, and $\mathcal H_K$ be the associated RKHS. Suppose that $g^*_\phi$ is measurable and bounded, $|g^*_\phi|\leq M_g$, and $|\sqrt \lambda_n b_n|\leq M_b$ almost surely for some constants $ M_g$ and $M_b$. In addition, we consider a sequence of tuning parameters $\lambda_n \rightarrow 0$ and $ n\lambda_n \rightarrow \infty$ as $n \rightarrow \infty$.
In order to study the excess risk bound of the $\phi$ loss, we need one additional assumption to bound the weight $W$.
\begin{assumption}
The outcome $Y$ is sub-Gaussian. Furthermore, we assume that $M_1<|\delta(L)|$, $M_2<f(Z=1|L)<1-M_2$ for some $0<M_1<1$, $0<M_2<1/2$ almost surely.
\label{asm:bounded}
\end{assumption}
Then, we have the following result.
\begin{theorem}
\label{thm:consistency}
Under Assumptions \ref{asm:unconfoundedness2}-\ref{asm:weak}, \ref{asm:bounded}, and further assume that
\begin{align*}
\sup_{l \in \mathcal L}|\widehat \delta(l) - \delta(l)| \overset{p}{\to} 0, \quad  \text{and}  \quad
\sup_{l \in \mathcal L}|\widehat f(z=1|l) - f(z=1|l)| \overset{p}{\to} 0,
\end{align*}
 as $n \rightarrow \infty$, we have the following convergence in probability,
\begin{align*}
\lim_{n\rightarrow \infty} R(g_{n}) = R^*,
\end{align*}
where $g_n =h_n +b_n$ is the estimated decision function from
\begin{align*}
\min_{g=h+b\in \cH_K+\{1\}} \frac{1}{n}\sum_{i=1}^{n} |\widehat w_i| \phi(\sign(\widehat w_i)a_ig(l_i)) +\frac{\lambda}{2}||h||_K^2.
\end{align*}
\end{theorem}
The proof is akin to \cite{zhao2012estimating,zhou2017augmented}.
The rate of convergence for the estimated treatment regime might also be studied under certain regularity conditions on the distribution of the data, such as the geometric noise assumption proposed by \cite{steinwart2007}.

\section{A locally efficient and multiply robust estimation of value function}
Consider the nonparametric model $\cM_{np}$ which places no restriction on the observed data law. Below, we characterize the efficient influence function of the value functional $\cV(\cD)$
in $\cM_{np}$ and therefore characterize the semiparametric efficiency bound for the model, where functional $\cV(\cD)$ is defined in Equation~\eqref{eq:vd}.

\begin{theorem}
\label{thm:triply1}
Under Assumptions~\ref{asm:unconfoundedness2}-\ref{asm:IV positivity} and \ref{asm:strong}, the efficient influence function of $\cV(\cD)$ in $\cM_{np}$ is given by
\begin{align*}
& EIF_{\cV(\cD)}=\frac{ZAYI\{A=\cD(L)\}}{ f(Z|L)\delta(L) }  - \bigg \{    \frac{ZE[AYI\{A=\cD(L)\}|Z,L] }{f(Z|L)\delta(L)}\\ & - \sum_z \frac{ zE[AYI\{A=\cD(L)\}|Z=z,L] }{\delta(L)} \\
  & + \frac{Z [A-E(A|Z,L) ] }{2f(Z|L)\delta(L)} \sum_z \frac{E[AYI\{A=\cD(L)\}|Z=z,L]z}{\delta(L)}  \bigg\}-\cV(\cD).
\end{align*}
Therefore, the semiparametric efficiency bound of $\cV(\cD)$ in $\cM_{np}$ equals $E[EIF_{\cV(\cD)}^2]$.
\end{theorem}

One cannot be confident that any of the required nuisance models to evaluate the efficient influence function can be correctly specified. It is of interest to develop a multiply robust estimation approach, which is guaranteed to deliver valid inferences about $\cV(\cD)$ provided that some but not necessarily all needed models are correct.
\yifan{When finite-dimensional models are used for nuisance parameters,
it is likely that all of them are misspecified leading to lack of consistency. Using infinite-dimensional
models can mitigate the problem, however, to achieve asymptotic linearity it is required that all the parts are
consistently estimated with sufficiently fast rates \citep{robins2017,Chernozhukov2018}.}

In order to describe our proposed multiply robust approach, consider the following three semiparametric models that place restrictions on different components of the observed data likelihood while allowing the rest of the likelihood to remain unrestricted.

\vspace{0.2cm}

\noindent $\cM_1$: models for $f(Z|L)$ and $\delta(L)$ are correct;

\vspace{0.4cm}

\noindent $\cM_2$: models for $f(Z|L)$ and  $\gamma(L) \equiv  \sum_z \{ zE[AYI\{A=\cD(L)\}|Z=z,L] \}/{\delta(L)} $ \\ \indent ~ are correct;

\vspace{0.4cm}

\noindent $\cM_3$: models for $\gamma(L), \gamma'(L) \equiv {E[AYI\{A=\cD(L)\}|Z=-1,L]}, \delta(L)$ and \\ \indent ~ $E[A|Z=-1,L]$ are correct.

\vspace{0.2cm}

Note that by Theorem~\ref{thm:gamma} presented in Section~\ref{sec:thm:gamma}, $\gamma(L)$ has the counterfactual interpretation $E[Y_{\cD(L)}|L]$,
 which may help formulate appropriate parametric models for the former. For instance, in case $Y$ is binary, $\gamma(L)$ would need to be specified with an appropriate link function to ensure it falls within the unit interval $(0,1)$.

Our proposed multiply robust estimator requires estimation of nuisance parameters $f(Z|L)$, $E(A|Z=-1, L)$, $\gamma'(L)$, $\delta(L)$ and  $\gamma(L)$. One may use maximum likelihood estimation for $f(Z|L)$, $E(A|Z=-1, L)$, $\gamma'(L)$, denoted as $\widehat f(Z|L)$, $\widehat E(A|Z=-1, L)$ and $\widehat \gamma'(L)$, respectively.

Because $\delta(L)$ and  $\gamma(L)$ are shared across submodels of the union model, i.e.,  $\cM_1 \cup \cM_3$, $\cM_2 \cup \cM_3$, respectively, in order to ensure multiple robustness, one must estimate these unknown functions in their respective union model.
For estimating $\delta(L)$, we propose to use doubly robust g-estimation \citep{robins1994,robins2000proceedings},
\begin{align*}
\PP_n \psi_1(L) \left[ A- \delta(L, \widehat \beta)\frac{1+Z}{2} - \widehat E(A|Z=-1, L)    \right] \frac{ Z }{\widehat f(Z|L)}=0,
\end{align*}
and we propose the following doubly robust estimating equation to estimate $\gamma(L)$,
\begin{align*}
& \PP_n \psi_2(L) \bigg[ {AYI\{A=\cD(L)\}}  -   {\widehat \gamma'(L) }  - \frac{[A-\widehat E(A|Z=-1,L) ]  \gamma(L,\widehat \eta)}{2} \bigg]  \frac{ Z }{\widehat f(Z|L)}= 0,
\end{align*}
\noindent where vector-valued functions $ \psi_1(L)$ and  $\psi_2(L)$ have the same dimension as $\widehat \beta$ and $\widehat \eta$, respectively. Thus, $\delta(L,\widehat \beta)$ is consistent and asymptotically normal in the union model $\cM_1 \cup \cM_3$, and $\gamma(L,\widehat  \eta)$ is consistent and asymptotically normal in the union model $\cM_2 \cup \cM_3$.
Similarly to results in \cite{tchetgen2018tr}, we have the following theorem.

\begin{theorem}
\label{thm:triply1.1}
Under Assumptions~\ref{asm:unconfoundedness2}-\ref{asm:IV positivity}, \ref{asm:strong} and standard regularity conditions,
\begin{align}
\widehat \cV_{MR}(\cD) =& \PP_n \Bigg [ \frac{ZAYI\{A=\cD(L)\}}{  \widehat f(Z|L)  \delta(L,\widehat \beta) }  -    \frac{Z \widehat \gamma'(L) }{\widehat f(Z|L)\delta(L,\widehat \beta) }  \nonumber \\ &+ \gamma(L,\widehat \eta)   - \frac{Z [A-\widehat E(A|Z=-1,L)] }{2\widehat f(Z|L)\delta(L,\widehat \beta)} \gamma(L,\widehat \eta) \Bigg ]
\label{eq:vmr}
\end{align}
is a consistent and asymptotically normal estimator of  $\cV(\cD)$ under the semiparametric union model $\cM_{union}=\cM_1 \cup \cM_2 \cup  \cM_3$. Furthermore, $\widehat \cV_{MR}(\cD)$ is semiparametric locally efficient in $\cM_{union}$ at the intersection submodel $\cM_{int}=\cM_1 \cap \cM_2 \cap \cM_3$.
\end{theorem}

Based on Theorem~\ref{thm:triply1.1}, one may evaluate the value function $E[Y_{\cD}(L)]$ for any given treatment regime $\cD$ with multiple robustness property. We derive the influence function of $\widehat \cV_{MR}(\cD)$ in Section~\ref{sec:mr2}, which can be used to nonparametrically estimate the standard deviation of value function under a given regime.

\section{Proof of Theorem~\ref{thm:triply1}}
\begin{proof}
In order to find the efficient influence function for $\cV(\cD)$,
we need to find the canonical gradient $G$ for $\cV(\cD)$ in the nonparametric model $\cM_{np}$, e.g, find a random variable $G$ with mean 0 and
\begin{align*}
\frac{\partial }{\partial t} \cV_t(\cD) \big|_{t=0}  = E[G S(\cO;t)]\big |_{t=0},
\end{align*}
where $S(\cO;t)={\partial \log f(\cO;t)}/{\partial t}$, and $\cV_t(\cD)$ is the value function under a regular parametric submodel in $\cM_{np}$ indexed by $t$ that includes the true data generating mechanism at $t=0$ \citep{vaart_1998}.
Note that we have
\begin{align*}
& \frac{\partial }{\partial t} \cV_t(\cD)\big |_{t=0} \\
= & E[ \frac{ZAYI(A=\cD(L))}{\delta(L)f(Z|L) } S(\cO)] - E \bigg[   \frac{ZAYI(A=\cD(L))}{\delta^2(L)f^2(Z|L) } [  \frac{\partial }{\partial t} f_t(Z|L) \delta(L) + \frac{\partial }{\partial t} \delta_t(L)f(Z|L)  ]    \bigg] \bigg |_{t=0}\\
= & (I) - (II) - (III).
\end{align*}
The second term
\begin{align*}
(II)= & E \bigg[   \frac{ZAYI(A=\cD(L))}{\delta^2(L)f^2(Z|L) }  \frac{\partial }{\partial t} f_t(Z|L) \delta(L) \bigg] \bigg |_{t=0}\\
= & E \bigg[   \frac{ZAYI(A=\cD(L))}{\delta(L)f(Z|L) }  S(Z|L)    \bigg]\\
= & E \bigg[   E[\frac{ZAYI(A=\cD(L))}{\delta(L)f(Z|L) }|Z,L]  S(Z|L)    \bigg]\\
= & E \bigg[  \Big \{ E[\frac{ZAYI(A=\cD(L))}{\delta(L)f(Z|L) }|Z,L]-   E[\sum_z \frac{zAYI(A=\cD(L))}{\delta(L)}|Z=z, L] \Big \}  S(Z|L)    \bigg]\\
= & E \bigg[  \Big \{ E[\frac{ZAYI(A=\cD(L))}{\delta(L)f(Z|L) }|Z,L]-   E[\sum_z \frac{zAYI(A=\cD(L))}{\delta(L)}|Z=z, L] \Big \}  S(Z,L)    \bigg]\\
= & E \bigg[  \Big \{ E[\frac{ZAYI(A=\cD(L))}{\delta(L)f(Z|L) }|Z,L]-   E[\sum_z \frac{zAYI(A=\cD(L))}{\delta(L)}|Z=z, L] \Big \}  S(\cO)    \bigg]\\
= & E \bigg[  \Big \{ \frac{ZE[AYI(A=\cD(L))|Z,L]}{\delta(L)f(Z|L) }-   \sum_z \frac{zE[AYI(A=\cD(L))|Z=z, L]}{\delta(L)} \Big \}  S(\cO)    \bigg].\\
\end{align*}

In order to calculate $(III)$, we need to calculate the term $\frac{\partial}{\partial t}\delta_t(L)$. To do so, we first calculate $\frac{\partial }{\partial t} E_t[A|Z=z,L]\big|_{t=0}$. Note that

\begin{align*}
& \frac{\partial }{\partial t} E_t[A|Z=z,L] \\
= & \frac{\partial }{\partial t} \int af_t(a|Z=z,L) da \\
= & \int  a \frac{\partial f_t(a|Z=z,L)}{f_t(a|Z=z,L)}f_t(a|Z=z,L) da  \\
= & E[A\frac{\partial f_t(A|Z=z,L)}{f_t(A|Z=z,L)} | Z=z,L],\\
\end{align*}
and
\begin{align*}
& \frac{\partial }{\partial t} E_t[A|Z=z,L]\bigg|_{t=0} \\
= & E[AS(A|Z=z,L) | Z=z,L ]\\
= & E[(A-E[A|Z=z,L]) S(A|Z=z,L)| Z=z,L ]\\
= & E[(A-E[A|Z=z,L]) S(A,Z=z|L) | Z=z,L ].
\end{align*}
Then
\begin{align*}
& 2\frac{\partial}{\partial t}\delta_t(L) \bigg |_{t=0}\\
= & E[(A-E[A|Z=1,L]) S(A,Z=1|L) | L ]- E[(A-E[A|Z=-1,L]) S(A,Z=-1|L) | L]\\
= & E[\frac{Z}{f(Z|L)} (A-E[A|Z,L]) S(A,Z|L)  |L]\\
= & E[\frac{Z}{f(Z|L)} (A-E[A|Z,L]) S(A,Z|L)  |L] + E[\frac{Z}{f(Z|L)} (A-E[A|Z,L]) S(L)  |L] \\
= & E[\frac{Z}{f(Z|L)} (A-E[A|Z,L]) S(A,Z,L)  |L] \\
= & E[\frac{Z}{f(Z|L)} (A-E[A|Z,L]) S(A,Z,L)  |L] + E[\frac{Z}{f(Z|L)} (A-E[A|Z,L]) S(Y|A,Z,L)  |L] \\
= & E[\frac{Z}{f(Z|L)} (A-E[A|Z,L]) S(A,Y,Z,L)  |L].\\
\end{align*}

It follows that
\begin{align*}
(III)= & E \bigg[   \frac{ZAYI(A=\cD(L))}{\delta^2(L)f^2(Z|L) } \frac{\partial }{\partial t} \delta_t(L)f(Z|L)  \bigg ] \bigg |_{t=0}\\
= & \frac{1}{2} E \bigg[  E [\frac{ZAYI(A=\cD(L))}{\delta^2(L)f(Z|L) } | L] E[\frac{Z}{f(Z|L)} (A-E[A|Z,L]) S(\cO)  |L] \bigg ]\\
= & \frac{1}{2} E \bigg[  \sum_z\frac{zE[AYI(A=\cD(L))| Z=z,L ]}{\delta^2(L) }  E[\frac{Z}{f(Z|L)} (A-E[A|Z,L]) S(\cO)  |L] \bigg ]\\
= & \frac{1}{2} E \bigg[  \sum_z\frac{zE[AYI(A=\cD(L))| Z=z,L ]}{\delta^2(L) }  \frac{Z}{f(Z|L)} (A-E[A|Z,L]) S(\cO)  \bigg ].
\end{align*}
Thus, $\frac{\partial }{\partial t} \cV_t(\cD)|_{t=0}$ further equals to
\begin{align*}
& E[ \frac{ZAYI(A=\cD(L))}{\delta(L)f(Z|L) } S(\cO)] \nonumber \\
 & - E \bigg[  \Big \{ \frac{ZE[AYI(A=\cD(L))|Z,L]}{\delta(L)f(Z|L) }-   \sum_z \frac{zE[AYI(A=\cD(L))|Z=z, L]}{\delta(L)} \Big \}  S(\cO)    \bigg]\\
& -\frac{1}{2} E \bigg[  \sum_z\frac{zE[AYI(A=\cD(L))| Z=z,L  ]}{\delta^2(L) }  \frac{Z}{f(Z|L)} (A-E[A|Z,L]) S(\cO)  \bigg ].
\end{align*}

So the canonical gradient in $\cM_{np}$ is
\begin{align*}
& \frac{ZAYI(A=\cD(L))}{ f(Z|L)\delta(L) }  - \bigg \{    \frac{ZE[AYI(A=\cD(L))|Z,L] }{f(Z|L)\delta(L)}\\ & - \sum_z \frac{ zE[AYI(A=\cD(L))|Z=z,L] }{\delta(L)} \\
  & +\frac{Z(A-E(A|Z,L) ) }{2f(Z|L)\delta(L)} \sum_z \frac{E[AYI(A=\cD(L))|Z=z,L]z}{\delta(L)}  \bigg\} -\cV(\cD).
\end{align*}
As shown by \cite{bickel1993efficient,newey1990,vaart_1998}, the canonical gradient in $\cM_{np}$ equals to the efficient influence function evaluated at observed data $\cO$, which completes our proof.

\end{proof}

\section{Proof of Theorem~\ref{thm:triply1.1}}\label{sec:mr2}
\begin{proof}
We start from multiply robustness.
Under some regularity conditions \citep{10.2307/1912526}, the nuisance estimators $\delta(L,\widehat \beta)$, $\gamma(L,\widehat \eta)$, $\widehat \gamma'(L)$, $\widehat f(Z|L)$, $\widehat E(A|Z=-1,L)$, converge in probability to $\delta(L,\beta^*)$, $\gamma(L,\eta^*)$, $\gamma^{*}{'}(L)$, $f^*(Z|L)$,  $E^*(A|Z=-1,L)$. It suffices to show that in the union model $\cM_{union}$,
\begin{align*}
& E\Bigg [ \frac{ZAYI(A=\cD(L))}{ f^*(Z|L)\delta(L,\beta^*) }  - \bigg \{    \frac{ Z\gamma^{*}{'}(L) }{f^*(Z|L)\delta(L,\beta^*) } - \gamma(L,\eta^*) \\
&    +\frac{Z [A-E^*(A|Z=-1,L)] }{2f^*(Z|L)\delta(L,\beta^*)}\gamma(L,\eta^*) \bigg\} \Bigg ] = E[Y_{\cD(L)}].
\end{align*}

We first note that
\begin{align}
& E\Bigg [ \frac{ZAYI(A=\cD(L))}{ f^*(Z|L)\delta(L,\beta^*) }  - \bigg \{    \frac{ Z\gamma^{*}{'}(L) }{f^*(Z|L)\delta(L,\beta^*) } - \gamma(L,\eta^*) \\
&+\frac{Z [A-E^*(A|Z=-1,L)] }{2f^*(Z|L)\delta(L,\beta^*)}\gamma(L,\eta^*) \bigg\} \Bigg ] \nonumber \\
= & E\Bigg [ \frac{ZAYI(A=\cD(L))}{ f^*(Z|L)\delta(L,\beta^*) }  - \bigg \{    \frac{Z [ \gamma(L, \eta^*)(1+Z)\delta(L,\beta^*) + 2 \gamma^{*}{'}(L) ] }{2f^*(Z|L)\delta(L,\beta^*) } - \gamma(L,\eta^*) \nonumber \\
&    +\frac{Z [A-E^*(A|Z=-1,L)-(1+Z)\delta(L,\beta^*)] }{2f^*(Z|L)\delta(L,\beta^*)}\gamma(L,\eta^*) \bigg\} \Bigg ]
\label{eq:mr}
\end{align}

If $\cM_1$ is correctly specified, Equation~\eqref{eq:mr} equals to
\begin{align*}
 & E\Bigg [ \frac{ZAYI(A=\cD(L))}{ f(Z|L)\delta(L) } - \frac{Z [A-E^*(A|Z=-1,L)-(1+Z)\delta(L)]}{2f(Z|L)\delta(L)}  \gamma(L,\eta^*) \Bigg ]\\
= & E\Bigg [ \frac{ZAYI(A=\cD(L))}{ f(Z|L)\delta(L) } - \frac{Z ( E[A|Z=-1,L] - E^*[A|Z=-1,L] ) }{2f(Z|L)\delta(L)}  \gamma(L,\eta^*)  \Bigg ]\\
= & E  \Bigg[ E\bigg [ \frac{ZAYI(A=\cD(L))}{ f(Z|L)\delta(L) }|L \bigg]\Bigg]  \\
= & E[Y_{\cD}(L)].
\end{align*}

If $\cM_2$ is correctly specified, Equation~\eqref{eq:mr} equals to
\begin{align*}
 & E\Bigg [ \frac{ZAYI(A=\cD(L))}{ f(Z|L)\delta(L,\beta^*) } - \frac{Z [A-E^*(A|Z=-1,L)-(1+Z)\delta(L,\beta^*)] }{2f(Z|L)\delta(L,\beta^*)} \gamma(L)\Bigg ]\\
= & E\Bigg [\sum_z \frac{z}{\delta(L,\beta^*)}  E[AYI(A=\cD(L))|Z=z,L] \\ &- \frac{Z [A- E^*(A|Z=-1,L)-(1+Z)\delta(L,\beta^*)] }{2f(Z|L)\delta(L,\beta^*)} \gamma(L)   \Bigg ]\\
= & E\Bigg [ \frac{\delta(L)}{\delta(L,\beta^*)}\gamma(L)   -\sum_z  \frac{z [  \delta(L)z- \delta(L,\beta^*)z ] }{2\delta(L,\beta^*)}\gamma(L)  \Bigg ]\\
= & E[Y_{\cD}(L)].
\end{align*}

Finally, if $\cM_3$ is correctly specified,
notice that
\begin{align*}
E[AYI(A=\cD(L))|Z,L]= \gamma(L){\delta(L)}\frac{1+Z}{2}+\gamma'(L),
\end{align*}
so we have Equation~\eqref{eq:mr} equals to
\begin{align*}
 & E \bigg[ \frac{ZAYI(A=\cD(L))}{ f^*(Z|L)\delta(L) }  - \bigg \{   \frac{Z [ \gamma(L)(1+Z)\delta(L) + 2 \gamma'(L) ] }{2f^*(Z|L)\delta(L) } - \gamma(L) \bigg\} \bigg]\\
= & E \bigg[ \frac{ZE[AYI(A=\cD(L))|Z,L]}{ f^*(Z|L)\delta(L) }  -   \frac{Z[\gamma(L) (1+Z)\delta(L)+2\gamma'(L)]}{2f^*(Z|L)\delta(L)} + \gamma(L) \bigg] \\
= & E[Y_{\cD}(L)].
\end{align*}

As shown by \cite{robins2001}, the efficient influence function in $\cM_{union}$ coincides with the efficient influence function in $\cM_{np}$, i.e., $EIF_{\cV(\cD)}$.
Thus, in order to show asymptotic normality and local efficiency, we need to derive the influence function of $\widehat \cV_{MR}(\cD)$.
Let $\eta$ be a vector including all nuisance parameters.
 From a standard Taylor expansion of $EIF_{\cV(\cD)}$ around $\cV(\cD)$ and $\eta$, following uniform weak law of large number \citep{NEWEY19942111} under some regularity conditions,
 we have
\begin{align*}
\sqrt n (\widehat \cV_{MR}(\cD) -\cV(\cD)) = \frac{1}{\sqrt n} \sum_{i=1}^{n} EIF_{\cV(\cD)}(\cO_i) + E( \frac{\partial EIF_{\cV(\cD)}}{\partial \eta} )\sqrt n(\widehat \eta - \eta) + o_p(1).
\end{align*}
Following the proof of multiple robustness, we have $E( \partial EIF_{\cV(\cD)}/\partial \eta)=0$ under the intersection model $\cM_{int}$.
This completes our proof.
\end{proof}

\section{Theorem~\ref{thm:gamma} and its proof}\label{sec:thm:gamma}

\begin{theorem}
Under Assumptions~\ref{asm:unconfoundedness2}-\ref{asm:IV positivity} and \ref{asm:strong}, we have that $\gamma(L)=E[Y_{\cD(L)}|L].$
\label{thm:gamma}
\end{theorem}
\begin{proof}
Note that
\begin{align*}
\gamma(L) = & \sum_z \frac{ zE[AYI(A=\cD(L))|Z=z,L] }{\delta(L)}\\
= & \frac{ E[AYI(A=\cD(L))|Z=1,L] }{\delta(L)} - \frac{E[AYI(A=\cD(L))|Z=-1,L] }{\delta(L)} \\
= & \sum_a \frac{ E[aY_aI(A=\cD(L))I(A=a)|Z=1,L] }{\delta(L)}\\
 & - \sum_a \frac{E[aY_aI(A=\cD(L))I(A=a)|Z=-1,L] }{\delta(L)} \\
 = & \sum_a \frac{ E\big[aE[Y_a|L,U]I(\cD(L)=a)\Pr(A=a|Z=1,L,U)\big] }{\delta(L)}\\
 & - \sum_a \frac{ E\big[aE[Y_a|L,U]I(\cD(L)=a)\Pr(A=a|Z=-1,L,U)\big] }{\delta(L)} \\
 = & \frac{ E\big[ E[Y_{1}|L,U]I(\cD(L)=1)\Pr(A=1|Z=1,L,U) \big]}{\delta(L)} \\
 &- \frac{ E\big[ E[Y_{-1}|L,U]I(\cD(L)=-1)\Pr(A=-1|Z=1,L,U)\big] }{\delta(L)} \\
 & - \frac{E\big[ E[Y_{1}|L,U]I(\cD(L)=1)\Pr(A=1|Z=-1,L,U)\big] }{\delta(L)} \\
 & + \frac{E\big[ E[Y_{-1}|L,U]I(\cD(L)=-1)\Pr(A=-1|Z=-1,L,U)\big] }{\delta(L)}\\
 = & E[Y_{1}|L]I(\cD(L)=1)+E[Y_{-1}|L]I(\cD(L)=-1)\\
 =& E[Y_{\cD(L)}|L].
\end{align*}
\end{proof}

\section{Additional simulations}
\subsection{Sensitivity analysis on the strength of the IV}
In this section, we conducted the sensitivity analysis on the strength of the IV.
Treatment $A$ was generated under a logistic regression with success probabilities,
\begin{align*}
\Pr(A=1|Z,L,U)=\expit\{2L^{(1)}+3Z-0.5U\},
\end{align*}
and
\begin{align*}
\Pr(A=1|Z,L,U)=\expit\{2L^{(1)}+2Z-0.5U\},
\end{align*}
respectively, with $Z$ a Bernoulli event with probability 1/2, and $U$ from a bridge distribution with parameter $\phi=1/2$.
The additive associations between $A$ and $Z$ defined as $\Pr(A=1|Z=1)- \Pr(A=1|Z=-1)$ are approximately equal to 0.74 and 0.54 for these two scenarios, respectively.
The additive association between $A$ and $Z$ is approximately equal to 0.66 for the scenario considered in the article.
Tables~\ref{simu3} and \ref{simu5} report the mean and standard deviation of value functions evaluated at estimated optimal regimes in test samples for two scenarios, respectively.
Tables~\ref{simu4} and \ref{simu6} report the mean and standard deviation of correct classification rates in test samples for two scenarios, respectively.
As can be seen from tables, higher compliance rate generally leads to a lower variance of the estimated regime in terms of both value functions and correct classification rates.

\begin{table}[!]
\begin{center}
\caption{\label{simu3}
Simulation results: Mean $\times 10^{-2} $ (sd $\times 10^{-2}$) of value functions}
\begin{tabular}{cccccc}
\noalign{\smallskip}
\noalign{\smallskip}
 & Kernel  & OWL & RWL & IV-IW   & IV-MR   \\
\noalign{\smallskip}
\multirow{2}{*}{1} & Linear &  94.5~(4.7)   &  96.5~(2.9) &  96.6~(3.7)  &  97.6~(2.1)      \\
                & Gaussian  &   87.7~(7.0)  &  95.0~(2.9) &  90.8~(8.2)  &  94.6~(4.8)    \\
\noalign{\smallskip}
\multirow{2}{*}{2} & Linear & 38.1~(18.9)  &  40.1~(19.3) & 93.0~(6.6) &  93.8~(6.0)      \\
                & Gaussian &  64.1~(9.0)   & 65.0~(9.2)  & 84.3~(10.4) & 87.7~(9.1)     \\
\noalign{\smallskip}
\noalign{\smallskip}
\multirow{2}{*}{3} & Linear &  354.0~(6.1) &  358.6~(3.0) & 359.5~(2.7)  &  359.6~(2.2)   \\
                & Gaussian   &  302.9~(33.3) &  356.4~(4.1)  & 320.9~(33.0)  &  357.0~(5.8)   \\
\noalign{\smallskip}
\multirow{2}{*}{4} & Linear &  275.4~(5.4)  & 275.5~(5.3) & 351.6~(8.4) &  351.7~(8.2)    \\
                & Gaussian  &  282.4~(14.2)  & 302.4~(14.6) & 314.0~(32.9) & 337.8~(19.1)     \\
\end{tabular}
\end{center}
\source{OWL:  outcome weighted learning; RWL: residual weighted learning; IV-IW: the proposed estimator with weight $\widehat W^{(1)}$ or $\widehat W^{(2)}$; IV-MR: the proposed multiply robust estimator with weight $\widehat W^{(1)}_{MR} $ or $\widehat W^{(2)}_{MR} $;  The empirical optimal value functions are 0.998, 0.995, 3.636, 3.630 for four scenarios, respectively.}
\end{table}

\begin{table}[!]
\begin{center}
\caption{\label{simu4}
Simulation results: Mean $\times 10^{-2}$ (sd $\times 10^{-2}$) of correct classification rates}
\begin{tabular}{cccccc}
\noalign{\smallskip}
\noalign{\smallskip}
 & Kernel  & OWL & RWL & IV-IW & IV-MR  \\
\noalign{\smallskip}
\multirow{2}{*}{1}  & Linear &  86.1~(5.6)   &  89.0~(4.1)  &  89.3~(4.7)  &  91.0~(3.3)     \\
                & Gaussian  &   79.2~(6.8)  &  86.9~(3.9) & 83.0~(8.1)  &  86.9~(5.4)    \\
\noalign{\smallskip}
\multirow{2}{*}{2} & Linear  &  43.3~(11.0) &   44.2~(10.9)  & 84.9~(7.2)  &  85.8~(6.7)   \\
                & Gaussian  &  59.4~(6.2)  &  60.3~(6.6) & 76.7~(9.2)  &  79.8~(8.2)    \\
\noalign{\smallskip}
\noalign{\smallskip}
\multirow{2}{*}{3} & Linear &  86.9~(4.4)  &  89.1~(4.0) & 91.3~(3.0)  &  90.6~(2.9)   \\
                & Gaussian  &  72.4~(9.4)  &  87.6~(4.3) & 77.8~(11.2)  &  89.3~(4.2)  \\
\noalign{\smallskip}
\multirow{2}{*}{4} & Linear &  37.3~(2.2)  &  37.4~(2.3) & 84.5~(5.9) &  84.3~(6.1)   \\
                & Gaussian  &  47.6~(7.1)  &  51.0~(8.3) & 71.0~(11.2) &  77.5~(9.5) \\
\end{tabular}
\end{center}
\source{OWL:  outcome weighted learning; RWL: residual weighted learning; IV-IW: the proposed estimator with weight $\widehat W^{(1)}$ or $\widehat W^{(2)}$; IV-MR: the proposed multiply robust estimator with weight $\widehat W^{(1)}_{MR} $ or $\widehat W^{(2)}_{MR} $.}
\end{table}


\begin{table}[!]
\begin{center}
\caption{\label{simu5}
Simulation results: Mean $\times 10^{-2} $ (sd $\times 10^{-2}$) of value functions}
\begin{tabular}{cccccc}
\noalign{\smallskip}
\noalign{\smallskip}
 & Kernel  & OWL & RWL & IV-IW   & IV-MR   \\
\noalign{\smallskip}
\multirow{2}{*}{1} & Linear &  96.8~(1.9)   &  97.7~(1.4) &  93.7~(6.5)  &  95.3~(4.9)      \\
                & Gaussian   &   87.9~(8.2)  &  96.0~(2.7) &  86.3~(10.6)  &  90.5~(8.0)   \\
\noalign{\smallskip}
\multirow{2}{*}{2} & Linear & 35.0~(17.1)  &  38.3~(18.7) & 88.8~(8.9) & 89.5~(9.0)      \\
                & Gaussian  &  59.0~(10.1)   & 58.8~(11.1) & 77.6~(12.6) & 81.1~(11.4)     \\
\noalign{\smallskip}
\noalign{\smallskip}
\multirow{2}{*}{3} & Linear &  358.0~(3.6) &  360.2~(2.1) & 356.8~(5.4)  &  357.4~(4.2)   \\
                & Gaussian   &  288.3~(33.5) &  357.5~(4.5)  & 305.7~(34.6)  &  350.9~(13.0)   \\
\noalign{\smallskip}
\multirow{2}{*}{4} & Linear &  274.6~(0.0)  & 275.0~(3.6) & 343.5~(20.0) &  345.1~(16.9)    \\
                & Gaussian    &  279.0~(11.2)  & 293.5~(13.2) & 297.9~(32.2) & 320.0~(25.1)   \\
\end{tabular}
\end{center}
\source{OWL:  outcome weighted learning; RWL: residual weighted learning; IV-IW: the proposed estimator with weight $\widehat W^{(1)}$ or $\widehat W^{(2)}$; IV-MR: the proposed multiply robust estimator with weight $\widehat W^{(1)}_{MR} $ or $\widehat W^{(2)}_{MR} $; The empirical optimal value functions are 0.998, 0.995, 3.636, 3.630 for four scenarios, respectively.}
\end{table}

\begin{table}[!]
\begin{center}
\caption{\label{simu6}
Simulation results: Mean $\times 10^{-2}$ (sd $\times 10^{-2}$) of correct classification rates}
\begin{tabular}{cccccc}
\noalign{\smallskip}
\noalign{\smallskip}
 & Kernel  & OWL & RWL  & IV-IW & IV-MR \\
\noalign{\smallskip}
\multirow{2}{*}{1} & Linear  &  89.4~(3.4)  &  91.1~(3.0) & 85.3~(7.0)  &  87.3~(5.6)   \\
                & Gaussian &  79.8~(8.0)  &  88.4~(3.8) & 78.5~(9.4)  &  82.5~(7.6)    \\
\noalign{\smallskip}
\multirow{2}{*}{2} & Linear &  41.4~(9.7) &   43.0~(10.2) & 80.2~(8.7)  &  80.9~(8.8)     \\
                & Gaussian  &  55.3~(6.4)  &  55.5~(7.2) & 71.1~(10.1)  &  73.8~(9.5)    \\
\noalign{\smallskip}
\noalign{\smallskip}
\multirow{2}{*}{3} & Linear &  90.0~(3.3)  &  91.0~(3.3) & 88.9~(4.1)  &  88.7~(3.6)   \\
                & Gaussian  &  69.2~(10.4)  &  88.9~(4.6) & 71.6~(11.8)  &  85.9~(6.5)  \\
\noalign{\smallskip}
\multirow{2}{*}{4} & Linear  &  37.0~(0.0)  &  37.2~(1.5) & 80.4~(9.7) &  80.8~(8.9)  \\
                & Gaussian  &  42.8~(5.4)  &  45.4~(6.5) & 65.5~(10.7) &  70.7~(10.0) \\
\end{tabular}
\end{center}
\source{OWL:  outcome weighted learning; RWL: residual weighted learning; IV-IW: the proposed estimator with weight $\widehat W^{(1)}$ or $\widehat W^{(2)}$; IV-MR: the proposed multiply robust estimator with weight $\widehat W^{(1)}_{MR} $ or $\widehat W^{(2)}_{MR} $.}
\end{table}

\subsection{Additional simulations with different sample sizes}
Tables~\ref{simu2501}-\ref{simu2502} and \ref{simu1001}-\ref{simu1002} report the simulation results with sample sizes 250 and 1000, respectively.

\begin{table}[!]
\begin{center}
\caption{\label{simu2501}
Simulation results: Mean $\times 10^{-2} $ (sd $\times 10^{-2}$) of value functions (sample size $n=250$)}
\begin{tabular}{cccccc}
\noalign{\smallskip}
\noalign{\smallskip}
 & Kernel  & OWL & RWL  & IV-IW   & IV-MR  \\
\noalign{\smallskip}
\multirow{2}{*}{1} & Linear &  92.0~(6.5)   &  95.5~(3.2) &  93.0~(6.2)  &  94.3~(5.4)      \\
                & Gaussian  &   82.7~(8.7)  &  92.8~(4.8) &  81.5~(11.4)  &  87.9~(9.8)   \\
\noalign{\smallskip}
\multirow{2}{*}{2} & Linear & 48.5~(22.7)  &  51.7~(22.3) & 87.0~(9.6) &  87.4~(9.5)      \\
                & Gaussian  &  61.6~(11.1)   & 62.0~(12.9) & 75.7~(12.1) & 77.4~(12.3)     \\
\noalign{\smallskip}
\noalign{\smallskip}
\multirow{2}{*}{3} & Linear &  349.5~(11.8) &  357.4~(4.0) & 354.7~(8.4)  &  355.9~(5.5)   \\
                & Gaussian   &  282.7~(35.4) &  354.0~(6.4) & 298.7~(35.7)  &  346.1~(17.2)    \\
\noalign{\smallskip}
\multirow{2}{*}{4} & Linear  &  278.5~(12.7)  & 280.8~(15.1) & 339.7~(21.1) &  340.3~(19.3)   \\
                & Gaussian  &  276.5~(17.7)  & 299.3~(15.4)  & 295.6~(33.7) & 317.8~(25.3)   \\
\end{tabular}
\end{center}
\source{OWL:  outcome weighted learning; RWL: residual weighted learning; IV-IW: the proposed estimator with weight $\widehat W^{(1)}$ or $\widehat W^{(2)}$; IV-MR: the proposed multiply robust estimator with weight $\widehat W^{(1)}_{MR} $ or $\widehat W^{(2)}_{MR} $;  The empirical optimal value functions are 0.998, 0.995, 3.636, 3.630 for four scenarios, respectively.}
\end{table}

\begin{table}[!]
\begin{center}
\caption{\label{simu2502}
Simulation results: Mean $\times 10^{-2}$ (sd $\times 10^{-2}$) of correct classification rates (sample size $n=250$)}
\begin{tabular}{cccccc}
\noalign{\smallskip}
\noalign{\smallskip}
 & Kernel   & OWL & RWL & IV-IW & IV-MR \\
\noalign{\smallskip}
\multirow{2}{*}{1}  & Linear  &  83.2~(7.0)   &  87.4~(4.6) &  84.3~(6.7)  &  86.0~(6.1)     \\
                & Gaussian  &   74.5~(8.1)  &  84.4~(5.4) & 74.2~(9.8)  &  80.1~(8.9)    \\
\noalign{\smallskip}
\multirow{2}{*}{2} & Linear &  50.0~(14.1) &   51.3~(13.4) & 78.3~(8.8)  &  78.7~(8.7)     \\
                & Gaussian &  58.0~(7.4)  &  58.5~(8.6) & 69.4~(9.6)  &  70.7~(9.8)    \\
\noalign{\smallskip}
\noalign{\smallskip}
\multirow{2}{*}{3} & Linear &  84.8~(5.3)  &  88.2~(4.3) & 87.6~(5.0)  &  87.6~(4.5)   \\
                & Gaussian  &  67.3~(9.8)  &  85.8~(5.5) & 70.2~(11.9)  &  83.3~(8.4)  \\
\noalign{\smallskip}
\multirow{2}{*}{4} & Linear &  39.0~(6.2)  &  39.9~(7.2) & 78.7~(9.6) &  78.6~(9.4)   \\
                & Gaussian  &  46.6~(6.9)  &  49.9~(8.6) & 65.1~(10.8) &  69.8~(10.6)  \\
\end{tabular}
\end{center}
\source{OWL:  outcome weighted learning; RWL: residual weighted learning; IV-IW: the proposed estimator with weight $\widehat W^{(1)}$ or $\widehat W^{(2)}$; IV-MR: the proposed multiply robust estimator with weight $\widehat W^{(1)}_{MR} $ or $\widehat W^{(2)}_{MR} $.}
\end{table}


\begin{table}[!]
\begin{center}
\caption{\label{simu1001}
Simulation results: Mean $\times 10^{-2} $ (sd $\times 10^{-2}$) of value functions (sample size $n=1000$)}
\begin{tabular}{cccccc}
\noalign{\smallskip}
\noalign{\smallskip}
 & Kernel  & OWL & RWL & IV-IW   & IV-MR   \\
\noalign{\smallskip}
\multirow{2}{*}{1} & Linear  &  97.7~(1.4)   &  98.4~(0.9)  &  97.6~(2.2)  &  98.2~(1.1)    \\
                & Gaussian  &   91.4~(6.0)  &  97.2~(1.6)  &  93.4~(5.3)  &  96.1~(2.6)  \\
\noalign{\smallskip}
\multirow{2}{*}{2} & Linear & 29.1~(8.6)  &  31.9~(12.8) & 95.0~(5.2) &  95.7~(4.5)      \\
                & Gaussian  &  62.5~(6.6)   & 62.1~(7.3) & 88.7~(8.1) & 90.8~(6.1)     \\
\noalign{\smallskip}
\noalign{\smallskip}
\multirow{2}{*}{3} & Linear &  359.4~(2.5) &  360.8~(1.5) & 360.4~(1.9)  &  360.1~(1.7)   \\
                & Gaussian  &  305.2~(29.8) &  359.2~(2.8)  & 331.8~(29.7)  &  359.0~(3.8)   \\
\noalign{\smallskip}
\multirow{2}{*}{4} & Linear  &  274.6~(0.0)  & 274.7~ (1.3) & 355.1~(4.2) &  354.7~(3.9)   \\
                & Gaussian   &  283.3~(10.4)  & 299.0~(11.5) & 325.9~(29.2) & 345.8~(13.1)    \\
\end{tabular}
\end{center}
\source{ OWL:  outcome weighted learning; RWL: residual weighted learning; IV-IW: the proposed estimator with weight $\widehat W^{(1)}$ or $\widehat W^{(2)}$; IV-MR: the proposed multiply robust estimator with weight $\widehat W^{(1)}_{MR} $ or $\widehat W^{(2)}_{MR}$; The empirical optimal value functions are 0.998, 0.995, 3.636, 3.630 for four scenarios, respectively.}
\end{table}

\begin{table}[!]
\begin{center}
\caption{\label{simu1002}
Simulation results: Mean $\times 10^{-2}$ (sd $\times 10^{-2}$) of correct classification rates (sample size $n=1000$)}
\begin{tabular}{cccccc}
\noalign{\smallskip}
\noalign{\smallskip}
 & Kernel   & OWL & RWL & IV-IW & IV-MR \\
\noalign{\smallskip}
\multirow{2}{*}{1}  & Linear &  91.0~(2.9)   &  92.8~(2.4) &  90.9~(3.4)  &  92.2~(2.6)      \\
                & Gaussian  &   83.4~(6.3)  &  90.3~(2.8) & 85.7~(5.5)  &  88.8~(3.5)    \\
\noalign{\smallskip}
\multirow{2}{*}{2} & Linear &  38.0~(4.5) &   39.4~(6.5) & 87.4~(5.9)  &  88.3~(5.3)     \\
                & Gaussian   &  57.6~(4.8)  &  57.5~(5.4) & 80.8~(7.4)  &  82.6~(6.1)   \\
\noalign{\smallskip}
\noalign{\smallskip}
\multirow{2}{*}{3} & Linear  &  91.3~(2.8)  &  91.5~(2.7) & 92.2~(2.5)  &  91.2~(2.4)  \\
                & Gaussian  &  73.8~(9.2)  &  90.7~(3.3) & 81.2~(10.5)  &  91.3~(3.0)  \\
\noalign{\smallskip}
\multirow{2}{*}{4} & Linear &  37.0~(0.0)  &  37.1~(0.7) & 87.0~(3.9) &  86.5~(3.7)   \\
                & Gaussian  &  44.1~(5.4)  &  47.9~(6.3) & 75.4~(10.7) &  81.5~(7.1)  \\
\end{tabular}
\end{center}
\source{OWL:  outcome weighted learning; RWL: residual weighted learning; IV-IW: the proposed estimator with weight $\widehat W^{(1)}$ or $\widehat W^{(2)}$; IV-MR: the proposed multiply robust estimator with weight $\widehat W^{(1)}_{MR} $ or $\widehat W^{(2)}_{MR} $.}
\end{table}

\bibliographystyle{asa}
\bibliography{msm,iv,owlsurvival,survtrees,survtrees2,causal}

\end{document}